\documentclass[11pt]{article}
\usepackage{amsmath,amsthm,latexsym,amssymb,amsfonts,epsfig, psfrag,color}
\usepackage{hyperref}
\usepackage{bbold}
\usepackage{graphicx}
\usepackage{subcaption}
\usepackage{sidecap}
\usepackage{wrapfig}
\usepackage{multirow}
\usepackage{authblk}
\usepackage{cite}

\usepackage{soul}

\usepackage[font={small}]{caption}

\makeatletter \@addtoreset{equation}{section} \makeatother

\newtheorem{theorem}{Theorem}[section]
\newtheorem{definition}[theorem]{Definition}
\newtheorem{lemma}[theorem]{Lemma}

\newtheorem{example}[theorem]{Example}
\newtheorem{observation}[theorem]{Observation}
\newtheorem{assumptions}[theorem]{Assumptions}

\def\be{\begin{equation}}
\def\ee{\end{equation}}
\def\ba{\begin{eqnarray}}
\def\ea{\end{eqnarray}}
\def\R{\mathbb{R}}

\newcommand\nn{\nonumber}

\def\tr{{\rm tr}}
\def\phys{{\rm phys}}

\def\g{{\cal G}}

\def\Nl{{\mathchoice
{\setbox0=\hbox{$\displaystyle\rm N$}\hbox{\hbox to0pt
{\kern0.4\wd0\vrule height0.9\ht0\hss}\box0}}
{\setbox0=\hbox{$\textstyle\rm N$}\hbox{\hbox to0pt
{\kern0.4\wd0\vrule height0.9\ht0\hss}\box0}}
{\setbox0=\hbox{$\scriptstyle\rm N$}\hbox{\hbox to0pt
{\kern0.4\wd0\vrule height0.9\ht0\hss}\box0}}
{\setbox0=\hbox{$\scriptscriptstyle\rm N$}\hbox{\hbox to0pt
{\kern0.4\wd0\vrule height0.9\ht0\hss}\box0}}}}

\def\Zl{{\mathchoice
{\setbox0=\hbox{$\displaystyle\rm Z$}\hbox{\hbox to0pt
{\kern0.4\wd0\vrule height0.9\ht0\hss}\box0}}
{\setbox0=\hbox{$\textstyle\rm Z$}\hbox{\hbox to0pt
{\kern0.4\wd0\vrule height0.9\ht0\hss}\box0}}
{\setbox0=\hbox{$\scriptstyle\rm Z$}\hbox{\hbox to0pt
{\kern0.4\wd0\vrule height0.9\ht0\hss}\box0}}
{\setbox0=\hbox{$\scriptscriptstyle\rm Z$}\hbox{\hbox to0pt
{\kern0.4\wd0\vrule height0.9\ht0\hss}\box0}}}}

\def\Ql{{\mathchoice
{\setbox0=\hbox{$\displaystyle\rm Q$}\hbox{\hbox to0pt
{\kern0.4\wd0\vrule height0.9\ht0\hss}\box0}}
{\setbox0=\hbox{$\textstyle\rm Q$}\hbox{\hbox to0pt
{\kern0.4\wd0\vrule height0.9\ht0\hss}\box0}}
{\setbox0=\hbox{$\scriptstyle\rm Q$}\hbox{\hbox to0pt
{\kern0.4\wd0\vrule height0.9\ht0\hss}\box0}}
{\setbox0=\hbox{$\scriptscriptstyle\rm Q$}\hbox{\hbox to0pt
{\kern0.4\wd0\vrule height0.9\ht0\hss}\box0}}}}

\def\Rl{{\mathchoice
{\setbox0=\hbox{$\displaystyle\rm R$}\hbox{\hbox to0pt
{\kern0.4\wd0\vrule height0.9\ht0\hss}\box0}}
{\setbox0=\hbox{$\textstyle\rm R$}\hbox{\hbox to0pt
{\kern0.4\wd0\vrule height0.9\ht0\hss}\box0}}
{\setbox0=\hbox{$\scriptstyle\rm R$}\hbox{\hbox to0pt
{\kern0.4\wd0\vrule height0.9\ht0\hss}\box0}}
{\setbox0=\hbox{$\scriptscriptstyle\rm R$}\hbox{\hbox to0pt
{\kern0.4\wd0\vrule height0.9\ht0\hss}\box0}}}}

\def\Cl{{\mathchoice
{\setbox0=\hbox{$\displaystyle\rm C$}\hbox{\hbox to0pt
{\kern0.4\wd0\vrule height0.9\ht0\hss}\box0}}
{\setbox0=\hbox{$\textstyle\rm C$}\hbox{\hbox to0pt
{\kern0.4\wd0\vrule height0.9\ht0\hss}\box0}}
{\setbox0=\hbox{$\scriptstyle\rm C$}\hbox{\hbox to0pt
{\kern0.4\wd0\vrule height0.9\ht0\hss}\box0}}
{\setbox0=\hbox{$\scriptscriptstyle\rm C$}\hbox{\hbox to0pt
{\kern0.4\wd0\vrule height0.9\ht0\hss}\box0}}}}

\def\Hl{{\mathchoice
{\setbox0=\hbox{$\displaystyle\rm H$}\hbox{\hbox to0pt
{\kern0.4\wd0\vrule height0.9\ht0\hss}\box0}}
{\setbox0=\hbox{$\textstyle\rm H$}\hbox{\hbox to0pt
{\kern0.4\wd0\vrule height0.9\ht0\hss}\box0}}
{\setbox0=\hbox{$\scriptstyle\rm H$}\hbox{\hbox to0pt
{\kern0.4\wd0\vrule height0.9\ht0\hss}\box0}}
{\setbox0=\hbox{$\scriptscriptstyle\rm H$}\hbox{\hbox to0pt
{\kern0.4\wd0\vrule height0.9\ht0\hss}\box0}}}}

\def\Ol{{\mathchoice
{\setbox0=\hbox{$\displaystyle\rm O$}\hbox{\hbox to0pt
{\kern0.4\wd0\vrule height0.9\ht0\hss}\box0}}
{\setbox0=\hbox{$\textstyle\rm O$}\hbox{\hbox to0pt
{\kern0.4\wd0\vrule height0.9\ht0\hss}\box0}}
{\setbox0=\hbox{$\scriptstyle\rm O$}\hbox{\hbox to0pt
{\kern0.4\wd0\vrule height0.9\ht0\hss}\box0}}
{\setbox0=\hbox{$\scriptscriptstyle\rm O$}\hbox{\hbox to0pt
{\kern0.4\wd0\vrule height0.9\ht0\hss}\box0}}}}

\newcommand{\cc}{\mathcal C}

\newcommand{\ch}{\mathcal H}

\newcommand{\s}{\mathcal{S}}

\def\nn{\nonumber}

\newcommand{\eqa}{\begin{eqnarray}}
\newcommand{\neqa}{\end{eqnarray}}

\usepackage{bbm}
\def\C{{\mathbbm C}}

\begin{document}

{\renewcommand{\thefootnote}{\fnsymbol{footnote}}

\title{An operational approach to spacetime symmetries:\\ Lorentz transformations from quantum communication}
\author[1]{Philipp A.\ H\"ohn\footnote{Present address: Institute for Quantum Optics and Quantum Information, Austrian Academy of Sciences, Boltzmanngasse 3, 1090 Vienna, Austria}\thanks{p.hoehn@univie.ac.at}}
\affil[1]{Perimeter Institute for Theoretical Physics, 31 Caroline Street North, Waterloo, ON N2L 2Y5, Canada}
\author[2,1,3]{Markus P.\ M\"uller\thanks{markus@mpmueller.net}}
\affil[2]{Department of Applied Mathematics, Department of Philosophy, University of Western Ontario, London, ON N6A 5BY, Canada}
\affil[3]{Institut f\"ur Theoretische Physik, Universit\"at Heidelberg, Philosophenweg 19, D-69120 Heidelberg, Germany}

}

\date{June 22, 2016}
\setcounter{footnote}{0}

\maketitle

\begin{abstract}
In most approaches to fundamental physics, spacetime symmetries are postulated a priori and then explicitly implemented in the theory.
This includes Lorentz covariance in quantum field theory and diffeomorphism invariance in quantum gravity, which are seen as
fundamental principles to which the final theory has to be adjusted. In this paper, we suggest, within a much simpler setting, that this kind of reasoning can actually be reversed,
by taking an operational approach inspired by quantum information theory. We consider observers in distinct laboratories, with local
physics described by the laws of abstract quantum theory, and without presupposing a particular spacetime structure.
We ask what information-theoretic effort the observers have to spend to synchronize their descriptions of local physics.
If there are ``enough'' observables that can be measured universally on several different quantum systems,
we show that the observers' descriptions are related by an element of the {orthochronous} Lorentz group ${\rm O^+}(3,1)$, together with a global scaling factor. Not only does this operational approach predict the Lorentz transformations, but it also accurately describes the behavior of relativistic Stern-Gerlach devices in the WKB approximation, and it correctly predicts that quantum systems carry Lorentz group representations of different spin. This result thus hints at a novel information-theoretic perspective on spacetime.
\end{abstract}

\section{Introduction}
Spacetime symmetries are powerful principles guiding the construction of physical theories. Their predictive power comes from the constraints
that arise from demanding that the fundamental physical laws are invariant under those symmetries. A paradigmatic example is given by quantum
field theory~\cite{Peskin}, where Lorentz covariance severely constrains the physically possible quantum fields and successfully predicts the different
types of particles that we find in nature. Similarly, diffeomorphism invariance guides our attempts to construct a quantum theory of gravity \cite{Rovelli:2004tv,Rovelli:1990ph,Rovelli:1990pi,Lewandowski:2005jk,Fleischhack:2004jc,Dittrich:2014ala}.

However, diffeomorphism invariance in particular reminds us of an important insight that played a fundamental role in the
construction of general relativity: that spacetime has a fundamentally operational basis {and that no external non-dynamical `background' exists on which physics takes place. Physical objects and systems can only be localized and referred to in {\it relation} to one another and not with respect to an extrinsic reference \cite{Rovelli:2004tv,Rovelli:1990ph,Rovelli:1990pi}}. The description of spacetime in terms of a metric field
can be derived from the equivalence principle, which is obtained as a statement about what an observer in a free-falling elevator can or cannot
get to know by performing measurements. Spacetime symmetries, like the Lorentz transformations in special relativity, can be interpreted as ``dictionaries''
that translate among distinct descriptions of the same physics given by different observers. These transformations constitute the {\it relations} among observers, and many of their properties can be understood
by analyzing operational tasks like clock synchronization.

All these operational protocols must be performed within the laws and limits of quantum physics. Thus, it is natural to take an ``inside point of view'' {on spacetime -- and physics in general --, giving primacy to relations among observers rather than spacetime itself. Furthermore, it suggests} to ask the fundamental question how different descriptions of the same physics by different observers are related, if we only assume the
validity of quantum theory. One might expect -- or at least hope -- that the spacetime symmetry group does not have to be postulated externally,
but instead emerges from the structure of physical objects themselves, that is, from the structure of quantum theory.

This general idea is not new, but has been pursued in several ways during the last decades. In his ``ur theory'', von Weizs\"acker~\cite{Weizsaecker}
argued that the spatial symmetries are inherited from the ${\rm SU}(2)$ symmetry of the quantum two-level system. Wootters~\cite{WoottersThesis,Wootters89} pointed out
the relation between distinguishability of quantum states and spatial geometry, and several authors~\cite{Holevo,Brody,MM3D,DakicBrukner3D} have analyzed aspects of the relation
between properties of space and the structure of quantum theory. This research has substantial overlap with questions in quantum information theory
regarding the use of quantum states as resources for reference frame agreement~\cite{Bartlett,GourSpekkens}.

In this paper, we approach this old question from a new angle by using the ideas and rigorous vocabulary of (quantum) information theory, {paired with a sometimes underappreciated insight from gravitational physics, namely that much of spacetime structure is encoded in the communication relations among all observers contained in it. Indeed, the communication relations encode the causal structure which, in turn, (under mild conditions) determines the spacetime geometry -- up to the conformal structure \cite{Hawking,Geroch}. One might even \emph{identify} spacetime with the set of all communication relations among all observers. Can we thus read out spacetime properties from elementary considerations on information communication? Rather than immediately attempting to derive the geometry, as a first step, we shall focus on one of the most basic ingredients underlying a spacetime picture, namely on reference frames and how to obtain the transformations among them from observer communication relations. }

To this end, we consider an information-theoretic
scenario involving two agents, Alice and Bob, {each with their own description of physics}, who cooperate to solve a simple communication task. The task is that Alice sends a request to Bob in terms of
classical information, and Bob is supposed to answer by sending the concrete physical object that was described by Alice's message.
Then we ask for the smallest possible group of transformations that Alice has to apply to her request to make sure the task succeeds.
We argue that this gives an operational definition of {the relation between reference frames. This yields an abstract and purely informational derivation of the reference frame transformations which does {\it not} rely on any externally given spatial or spacetime symmetries. We neither presuppose any specific spacetime or causal structure that Alice and Bob are `immersed' in, nor do we assume a Lorentzian or Euclidean signature and neither a specific dimension. Nevertheless, we shall argue that these reference frame transformations must be ultimately related to the local spacetime isometry group.}

After proving some general properties of this scenario in Section~\ref{SecReferenceFrames}, we consider in Section~\ref{SecO3} the special case that the objects to
transmit are stand-alone finite-dimensional quantum states. Under the hypothesis that many different sorts of quantum systems can be measured in the same
``universal'' devices, we show that the resulting transformation group, relating Alice's and Bob's reference frames, will be the orthogonal group ${\rm O}(3)$.
This confirms von Weizs\"acker's intuition~\cite{Weizsaecker}, but puts it on firm operational grounds, by specifying detailed physical background assumptions
that are sufficient (and probably necessary) to arrive at this conclusion, and by pointing out how these assumptions are realized in actual physics. Furthermore, we also derive the fact
that quantum systems carry projective representations of ${\rm SO}(3)$ of different spin.

In Section~\ref{sec_phys}, we argue that the previous derivation contained one crucial oversimplification: namely that the outcomes of measurements (that is, the eigenvalues of observables) are simply abstract labels without physical meaning. Dropping this assumption, and considering the possibility that the outcomes are actual physical quantities, leads us to consider different \emph{types} of quantum systems. Analyzing the communication task in this more general setting, it turns out
that Alice's and Bob's descriptions will be related by a Lorentz transformation, an element of ${\rm O}^+(3,1)$. These transformations act as isometries among
{\it different} finite-dimensional Hilbert spaces.

Thus, we arrive at the Lorentz group, without having assumed any aspects of special relativity in the first place. Furthermore, the resulting formalism turns out to
correctly describe the behavior of relativistic Stern-Gerlach measurement devices, under a WKB approximation where the spin degree of freedom does not
mix with the momentum of the particle. This provides further evidence for a spacetime interpretation of this abstractly derived Lorentz group.

In this article, we therefore reverse the standard point of view: instead of postulating a symmetry group and working out its physical or information-theoretic consequences (as is usually done e.g.\ in quantum field theory or the theory of quantum reference frames~\cite{Bartlett,GourSpekkens}), we start with natural information-theoretic and physical assumptions and derive the symmetry group from them. The motivation is to develop a novel information-theoretic understanding of spacetime structure, beginning here with an informational perspective on reference frames and their relations.

\section{Two agents who have never met: an operational account of reference frames}
\label{SecReferenceFrames}
Reference frames are usually introduced in theoretical physics in situations where there is some geometric structure which can be described in different
coordinate systems. The fact that these different descriptions are relevant has however an operational interpretation: \emph{different observers (agents, or physicists)
will in general use different descriptions for the same physical property}, as long as they have not agreed on a method of description beforehand. {The origin of this disagreement may be physical or simply the use of distinct conventions by different observers.}

In the following, we will emphasize this operational point of view by taking an approach that is well-known from information theory: we describe a
scenario involving two parties, traditionally called Alice and Bob, who cooperate to solve a certain communication task. What is to be communicated in our case,
however, will not be classical information, but different kinds of physical objects.
The case where the physical objects in question are abstract quantum states, carrying some group representation for which Alice and Bob may not share
a common frame of reference, is a standard scenario in quantum information theory, cf.~\cite{Bartlett,GourSpekkens} and the references therein. A standard question
asked there is, for example, how many asymmetric quantum states are needed to serve as a replacement of a classical reference frame. This usually presupposes the {\it classical} symmetry group which is to translate among distinct quantum reference frames. That is, the quantum reference frames have to be adjusted to the assumed classical symmetry group.

However, this perspective is conceptually not fully satisfactory because the classical world emerges from a quantum one, rather than the other way around. In a fundamentally quantum world, we better reverse the perspective and {\it derive} the appropriate classical symmetry group from quantum structures. That is, we ought to be able to explain the quantum origin of the classical symmetry group, rather than assuming the latter and adjusting the quantum theory to it. This shall be the goal of the present manuscript.

Our approach is therefore substantially different from the literature on quantum reference frames in that it asks a distinct kind of question: \emph{given some sort of physical objects (for example, quantum systems), what is the \underline{smallest}
possible transformation group that translates between Alice's and Bob's descriptions, given that they both cooperate and use all available physical structure?}
That is, we are not assuming any externally given group that acts on the objects, but instead ask for the group that emerges from the structure of the
physical objects themselves.

Consider an observer -- Alice -- who has access to a certain property of a physical system (in the sense that she can measure it), and who tries to describe this
property in mathematical terms. Denote by $\s_\phys$ all the values (or states, or configurations) that this physical property can possibly attain; then Alice is looking for
a map $\varphi_A:\s_\phys\to\s$, where $\s$ is the set of mathematical objects that she uses to describe the physical property.

As a simple example, suppose that Alice lives in three-dimensional Euclidean space of classical physics, and would like to describe the velocity of a billiard ball
at a given time. She can determine the velocity by a measurement, and read off a description of the velocity from a scale on her measurement device.
In this case, $\s_\phys$ denotes the set of all possible physical velocities, and $\s$ equals $\Rl^3$, the set of vectors with three real numbers as entries.

Different observers may (and will) in general use different maps to encode physics into mathematical descriptions; if we have another observer, Bob, then he
may use his own map $\varphi_B$, with $\varphi_B\neq \varphi_A$. Given a fixed physical property $x\in\s_\phys$, Alice's description will be $\varphi_A(x)$,
while Bob's description will be $\varphi_B(x)$. There are many different possible reasons for this. One possible reason is that Alice and Bob are using different
kinds of measurement devices to determine the physical property, for example differing in the choice of units or orientation.

Ideally, we would like to postulate that $\varphi_A,\varphi_B:\s_\phys\to\s$ are
one-to-one and onto (that is, bijective), at least in principle.
In this case, the map $\varphi_A\circ \varphi_B^{-1}$ is well-defined, and yields Alice's description
in terms of Bob's. Similarly, $T:=\varphi_B\circ\varphi_A^{-1}$ computes Bob's description from Alice's. These are maps from $\s$ to $\s$ (assuming that both observers use
the same mathematical objects to encode the physical property), and they are bijective because $\varphi_A$ and $\varphi_B$ are.
They are elements of the group of bijections of $\s$ into itself, and they play a natural role in a simple communication task involving Alice and Bob which is
depicted in Figure~\ref{fig_scenario1}.
\begin{figure*}[!hbt]
\begin{center}
\includegraphics[angle=0, width=8cm]{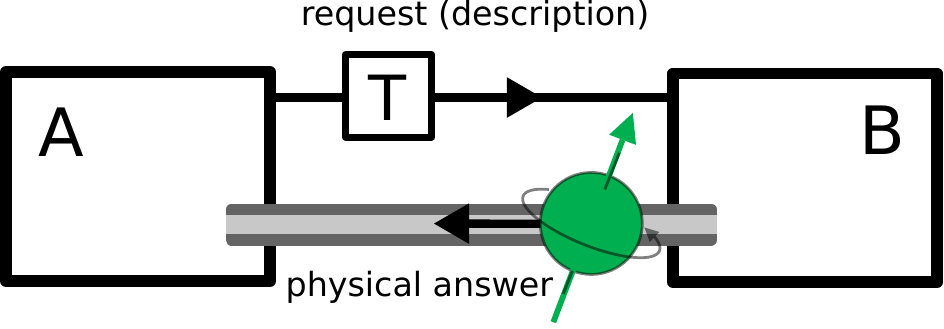}
\caption{The communication task. Alice requests a physical object by sending a classical description, and Bob answers by sending the corresponding
physical system in the desired state. If Alice and Bob use different maps $\varphi_A$ and $\varphi_B$ to describe physical objects, a correcting transformation $T$
must be applied to Alice's message before it arrives at Bob's laboratory. This $T$ is seen as the ``information gap'' between Alice and Bob; agreeing on $T$ corresponds to agreeing on a common
frame of reference with regards to the desired physical object.}
\label{fig_scenario1}
\end{center}
\end{figure*}

The scenario assumes that Alice is able to send classical information to Bob (as a colorful illustration, we may think of Alice and Bob being able to talk on the telephone).
In the first step of the game, Alice sends
Bob a ``classical request'', telling him to prepare a certain physical system. In the second step, Bob answers this request by sending an actual instance of
the requested physical system with the desired property.

Clearly, if Alice and Bob have never met before, and thus have never agreed on a shared reference frame, Bob will not know what map $\varphi_A$ Alice is
using to describe the desired physical quantity. The best he can do is to guess $\varphi_A$ -- or, equivalently, hope for the unlikely fact that $\varphi_A=\varphi_B$ -- and
send a physical system $x$ such that $\varphi_B(x)$ equals the description that Alice sent. If Alice checks whether she received the desired physical object, she will in
general detect failure, except for the unlikely case that $\varphi_A=\varphi_B$.

To correct for this problem, Alice can apply a correcting transformation $T$ before the classical information is actually sent to Bob.
It is easy to see that the protocol succeeds if Alice applies $T=\varphi_B\circ \varphi_A^{-1}$ to the classical description she sends out -- the problem is, of course,
that Alice may not know $\varphi_B$ and thus does not know $T$. However, if she knows $\varphi_B$, she can make the protocol succeed even if Bob does
not know $\varphi_A$. Regardless of whether Alice knows $\varphi_B$ or not, we may regard $T$ as the \emph{relation between Alice's and Bob's local descriptions}.
Not knowing $T$ may be regarded as an ``information gap'' between them; closing the information gap, i.e.\ negotiating on a map $T$, corresponds to setting
up a common frame of reference (regarding the physical property that is intended to be sent).

Sometimes, instead of preprocessing her classical description via the map $T=\varphi_B\circ\varphi_A^{-1}$, Alice may equivalently postprocess the physical
system that she obtains by applying the physical transformation $T_\phys:\s_\phys\to\s_\phys$, given by $T_\phys:=\varphi_A^{-1}\circ \varphi_B$, cf.\ Figure~\ref{fig_scenario2}.
If this is possible, we call $T$ \emph{implementable}. Note that $T$ and $T_\phys$ are different kinds of maps: $T$ is a transformation on \emph{classical descriptions},
and thus something that can always be accomplished by doing computations on representations on a piece of paper (or in a computer memory). The map $T_\phys$,
in contrast, is an actual physical transformation; in many cases, physics may disallow the actual implementation  of $T_\phys$. We will soon discuss an example below.
\begin{figure*}[!hbt]
\begin{center}
\includegraphics[angle=0, width=8cm]{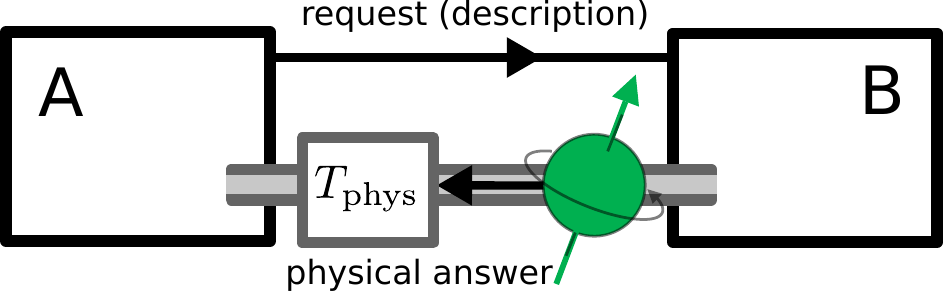}
\caption{Sometimes, instead of preprocessing the classical request via $T$, Alice may equivalently postprocess the physical answer via $T_\phys$. In this case,
we call $T$ \emph{implementable}.}
\label{fig_scenario2}
\end{center}
\end{figure*}

If there were no further restrictions on the maps $\varphi_A,\varphi_B$,
then the coordinate transformation functions $T=\varphi_B\circ\varphi_A^{-1}$ could exhaust all possible bijective maps from $\s$ to itself.
However, in many cases, it is physically impossible to actually use every map $\varphi$ as a description map, in the sense that $\varphi(x)$ may be physically
impossible to determine, given the physical object $x$. This is the case when there are physical
restrictions on the possible measurement devices that agents may access or implement. As a simple example, if we have again Alice in three-dimensional classical
mechanics, then only \emph{continuous} maps $\varphi_A$ will be physically relevant, simply because of the possible measurement procedures that Alice
may implement: she can build devices that measure velocity to arbitrary accuracy, but never to perfect accuracy. But then, the coordinate transformation maps $T=\varphi_A\circ\varphi_B^{-1}$ (and $T^{-1}=\varphi_B\circ\varphi_A^{-1}$)
must be continuous, too. In general, we obtain a group $\g_{\max}$ of physically meaningful transformations,
\[
   \g_{\max}=\left\{T:\s\to\s\,\,|\,\,T \mbox{ transforms between physically accessible/allowed descriptions}\right\}.
\]
In the example of classical physics and billiard ball velocity, this would be the group of homeomorphisms of $\Rl^3$ into itself. In general, one would
expect that $\g_{\max}$ is the subgroup of all bijections that preserve some key physical structure (which, in this example, is the topology).

Now recall the scenario from Figure~\ref{fig_scenario1}. Suppose that Alice and Bob are both cooperative and would like to
agree on a transformation $T$ in order to succeed with the communication task, and do so with as little effort as possible.
For this, they are willing to adjust their own local descriptions as long as it facilitates their effort to find the relation between their respective descriptions of the observed system. Then one way to reduce the effort is to stop using \emph{arbitrary} possible encodings
$\varphi_A$ and $\varphi_B$, and to draw instead from a \emph{physically distinguished subset} of encodings.

\begin{definition}
\label{DefMinimal}
A subgroup $\g\subset\g_{\max}$ is called \emph{achievable} if there is a subset of ``distinguished encodings'' $\Phi$ such that every observer can locally ensure by physical means that she is using a distinguished encoding $\varphi\in\Phi$, and
\[
   \g=\left\{\varphi_1\circ\varphi_2^{-1}\,\,|\,\, \varphi_1,\varphi_2\in\Phi\right\}.
\]
An achievable subgroup $\g$ is called \emph{minimal} if there is no proper subgroup $\g'\subsetneq\g$ which is achievable, too.
\end{definition}
Why is the set $\g$ always a group? This follows from a natural assumption on any subset $\Phi$ of physically distinguished
encodings; namely that
\begin{equation}
   \varphi_1,\varphi_2,\varphi_3\in\Phi\enspace\Rightarrow\enspace \varphi_3\circ\varphi_2^{-1}\circ\varphi_1\in\Phi.
   \label{eqGroup}
\end{equation}
Intuitively this is plausible: if $\varphi_1,\varphi_2,\varphi_3$ are ``equally natural'' encodings of physical objects, then an agent can obtain
another physically distinguished encoding by first using encoding $\varphi_1$, then determining the physical object that would have the same description via $\varphi_2$,
and finally encode this object via $\varphi_3$, and the resulting encoding should not be outside of the set of natural encodings. More formally, we can argue
as follows. In our communication task, Alice and Bob will try to use a common set of encodings $\Phi$ that is \emph{as small as possible}. If~(\ref{eqGroup})
was violated, then they could use a smaller set $\tilde\Phi$ of encodings, which would facilitate their task, rendering the choice of $\Phi$ inefficient.
To see this, suppose that there are $\varphi_1,\varphi_2,\varphi_3\in\Phi$ such that $\varphi_3\circ\varphi_2^{-1}\circ \varphi_1\not\in\Phi$.
Set $f:=\varphi_3\circ\varphi_2^{-1}$, which is a fixed map from $\s$ to $\s$ that Alice can describe to Bob in terms of classical information,
and define $\tilde\Phi:=\{\varphi\in\Phi\,\,|\,\, f\circ\varphi\not\in\Phi\}$.
Since $\varphi_1\in\tilde\Phi$ but $\varphi_2\not\in\tilde\Phi$, the set $\tilde\Phi$ is a non-empty strict subset of $\Phi$ which Alice and Bob
could use instead to encode physical objects. Thus, it makes sense to assume~(\ref{eqGroup}), from which it easily follows that $\g$ is closed with respect to composition
and inversion, that is, $\g$ is a group.

As a simple example, recall the example of the billiard ball velocity in classical mechanics. Instead of using an arbitrary
continuous encoding map $\varphi_A$, Alice may choose to use an \emph{inertial frame} to describe the velocity. Classical mechanics teaches us
how to achieve this: an inertial frame is one in which Newton's laws (stated in their inertial form) are valid; this is something that Alice can check by physical means.

Different inertial frames are related by a Galilean transformation. Since we are not interested in the spatial location of the billiard ball
(it will be part of the protocol to make sure it is located in Alice's lab after being sent by Bob), only the action of this group on the ball's momentum
is relevant. Restricting the Galilean group to the ball's velocity, we obtain the \emph{Euclidean group} ${\rm E}(3)$, describing combinations
of rotations and translations in momentum space which relate Alice's and Bob's laboratory.
Moreover, this group is also \emph{minimal}: there is clearly no way in which a
proper subgroup of ${\rm E}(3)$ would be sufficient to relate the velocity descriptions of observers who have never met. This follows from the fact that every element of ${\rm E}(3)$ defines a distinct inertial frame for velocity (relative to some reference inertial frame), and all inertial frames are physically equivalent.

What would be examples of non-minimal achievable groups in this case? The simplest example is $\g_{\max}$, which in this situation is the group of all homeomorphisms of $\Rl^3$
itself. We obtain it by allowing all continuous encodings $\varphi$; every observer can make sure to use a continuous encoding by simply \emph{being
forced to do this} by physics, as discussed before.

A less trivial example would be given by allowing all \emph{affine-linear} invertible encodings $\varphi$; that is,
encodings that preserve the vector space structure, but not necessarily the distance between points in momentum space. The corresponding group $\g'$ would
be ${\rm Aff}(3,\Rl)$, the affine group. Having Alice and Bob agree on an element of this group to establish a common frame of reference is more efficient than
referring to the full group $\g_{\max}$; it would describe a situation where Alice and Bob are not using arbitrary encodings, but those that they can find by
probing the affine vector space structure of momentum space (and building their measurement devices accordingly). If Alice and Bob are for some reason
not capable of measuring angles or lengths (notions that are only computable from an inner product), this would be the best they can do.

However, we know that in principle they \emph{can} do better, which is reflected by the fact that $\g'$ is achievable, but not minimal. It contains a
proper subgroup which is also achievable -- which is the Euclidean group. The Euclidean group corresponds to the best possible strategy to agree on a frame of reference, in the
sense that different observers can \emph{individually} commit to using inertial frame encodings, and by doing so, minimize the information gap between them.

Is there always a \emph{unique} minimal group? Taken literally, the answer is ``no'', as can be seen in our standard example. Let $\Phi$ be the set
of all encodings of velocities into vectors in $\Rl^3$ which correspond to choices of inertial frames, such that the corresponding group
$\g=\{\varphi_1\circ\varphi_2^{-1}\,\,|\,\, \varphi_1,\varphi_2\in\Phi\}$ is the Euclidean group. Consider the function $f(x_1,x_2,x_3)=(x_1^3,x_2^3,x_3^3)$,
which is a bijective map from $\Rl^3$ to itself. Then $\Phi':=\{f\circ\varphi\,\,|\,\,\varphi\in\Phi\}$ is also a physically distinguished encoding.
It is an encoding where an observer chooses an arbitrary inertial frame, determines the velocity in the corresponding coordinates, and then
takes the third power of all entries. The corresponding group $\g'$ is then achievable and even \emph{minimal} according to Definition~\ref{DefMinimal},
and it is different from the Euclidean group.

However, $\g'$ is \emph{isomorphic} to the Euclidean group; in fact, we have $\g'=f\circ\g\circ f^{-1}$. Taking any element $f\in\g$ and mapping it to $g'=f\circ g \circ f^{-1}$
is a bijective group homomorphism. That is, $\g'$ is ``exactly the same'' as $\g$ except for a relabeling. As an abstract group, the Euclidean group is the
unique minimal group for the scenario we described. Interestingly, this turns out to be true in \emph{all} scenarios:
\begin{lemma}
For any given scenario, all achievable minimal groups are isomorphic. That is, for any given scenario, the group $\g_{\min}$
which describes the minimal effort that two observers have to spend to agree on a common frame of reference is unique as an abstract group.
\end{lemma}
\begin{proof}
Fix any scenario.
Suppose that $\g$ and $\mathcal{H}$ are both achievable and minimal according to Definition~\ref{DefMinimal}. Denote by $\Phi_{\g}$ and $\Phi_{\mathcal{H}}$ the
corresponding sets of physically distinguished encodings. Choose $\varphi_{\g}\in\Phi_{\g}$ and $\varphi_{\mathcal{H}}\in\Phi_{\mathcal{H}}$ arbitrarily.
Set $f:=\varphi_{\g}\circ\varphi_{\mathcal{H}}^{-1}$, then $f:\s\to\s$ is bijective. Define $\Phi'_{\mathcal{H}}:=f\circ \Phi_{\mathcal{H}} = \{f\circ\varphi\,\,|\,\, \varphi
\in\Phi_{\mathcal{H}}\}$. Then $\Phi'_{\mathcal{H}}$ is a set of physically distinguished encodings. The corresponding group is
\[
   \mathcal{H}'=\{\varphi_1\circ\varphi_2^{-1}\,\,|\,\, \varphi_1,\varphi_2\in\Phi'_{\mathcal{H}}\}
   =\{f\circ \varphi_1\circ \varphi_2^{-1}\circ f^{-1}\,\,|\,\, \varphi_1,\varphi_2\in\Phi_{\mathcal{H}}\} = f\circ\mathcal{H}\circ f^{-1}.
\]
Thus, the group $\mathcal{H}'$ is isomorphic to the group $\mathcal{H}$. Furthermore, $\mathcal{H}'$ is achievable and minimal, because $\mathcal{H}$ is.
Note that $\varphi_{\g}=f\circ\varphi_{\mathcal{H}}$, hence $\varphi_{\g}\in\Phi'_{\mathcal{H}}$. Thus, $\Phi:=\Phi_{\g}\cap\Phi'_{\mathcal{H}}$ is not empty.
Furthermore, it is a set of physically distinguished encodings (observers can physically ensure that an encoding is in $\Phi_{\g}$ and also that it is in $\Phi'_{\mathcal{H}}$,
hence they can ensure that it is in both). Therefore, the corresponding group
\[
   \mathcal{K}:=\{\varphi_1\circ \varphi_2^{-1}\,\,|\,\, \varphi_1,\varphi_2\in\Phi\}
\]
is achievable, and we have $\mathcal{K}\subset\g$ as well as $\mathcal{K}\subset\mathcal{H}'$. Since $\g$ and $\mathcal{H}'$ are both minimal,
we must have $\mathcal{K}=\g$ and $\mathcal{K}=\mathcal{H}'$, hence $\g=\mathcal{H}'$, and so $\g$ and $\mathcal{H}$ are isomorphic.
\end{proof}

Proper Euclidean transformations (in the connected component of the identity) are implementable:
instead of preprocessing her classical description, Alice can simply rotate and accelerate the physical objects she obtains from Bob.
The question whether a spatial reflection is implementable or not depends on the detailed assumptions on the physics. For example, we can think of a machine
that measures the velocity $v$ of an object, and then accelerates the object exactly to velocity $-v$. If the physical background assumptions allow this machine,
then reflections will be implementable, too. In the following, we will consider transformations on quantum states, where implementability is more severely restricted
by the probabilistic structure.

In general, there may be scenarios in which no minimal achievable group exists at all; however, all examples of this that we can think of at present are rather unphysical,
for example in the sense that they refer to groups which are not topologically closed. In all remaining examples of this paper, a minimal group $\g_{\min}$ exists. {A priori this $\g_{\min}$ may depend on the communicated objects. The smallest possible group $\g_{\min}$ which translates among the description of {\it all} physical objects that agents may communicate embodies the minimal efforts they have to spend in order to agree on a description of physics. Accordingly, we shall later take it as {\it defining} the reference frame transformations.}

We finish this section with an example meant to elucidate in more detail the rules of the communication task that we have in mind.
Suppose that the physical objects to be transmitted are \emph{pairs of vectors} in classical mechanics; say, the velocities of \emph{two} billiard balls.
Since the Euclidean group ${\rm E}(3)$ is a (minimal) achievable group for \emph{one} billiard ball, the group
$\g={\rm E}(3)\times {\rm E}(3)$ is achievable -- two copies of the group, one for each billiard ball. However, Alice and Bob can just agree to use the \emph{same}
encoding (resp.\ inertial frame) for the description of \emph{both} billiard balls, and thus achieve $\g_{\min}={\rm E}(3)$ also for pairs of billiard balls.
This is possible whenever it is physically clear that if Bob sees two billiard balls next to each other and acknowledges that they have the same velocity,
then Alice will agree with this fact even after the balls have been transported to her laboratory.\footnote{This tacitly presumes the billiard balls to remain invariant under transport.}

This example illustrates that \emph{objective relations between
given physical objects} represent useful physical structure that can be used to simplify the communication task. This will become important in Subsections~\ref{SubsecFullLab}
and~\ref{Subsec44}.

\section{Abstract quantum states and the orthogonal group ${\rm O}(3)$}
\label{SecO3}
\subsection{Transmitting finite-dimensional quantum states}
\label{SubsecTransmitFinite}
We now consider the special case that the objects Alice and Bob are intending to describe and to send are \emph{abstract finite-dimensional
quantum states}\footnote{We could also consider normal quantum states on infinite-dimensional separable Hilbert spaces. Many conclusions of
Section~\ref{SecO3} should apply to these as well.}.
In the communication scenario in Figure~\ref{fig_scenario1}, Alice sends a request to Bob, asking him to prepare
a physical system in a specific $N$-level quantum state; Bob in turn sends a system in this state to Alice.
We are only interested in the quantum states themselves, not in the question whether these quantum
systems are correlated with any other system. In particular, if Bob is supposed to prepare and send a mixed quantum state, it does not matter whether this
is a proper or improper mixture, and whether there is any other system that serves as a purifying system of the mixed state. All that matters for the communication
task to succeed is the pure or mixed state that will in the end arrive in Alice's laboratory, which should match Alice's request.

The ``information gap'' in this scenario is that Alice and Bob will in general not have agreed beforehand on a common orthonormal basis in the Hilbert
space that is supposed to carry the quantum state. Since different bases are related by unitaries, the following result is not surprising:
\begin{lemma}
\label{LemQuantumGeneral}
The projective unitary antiunitary~\cite{Parthasarathy} group ${\rm PUA}(N)$ is achievable in the setup described above -- that is, the group of conjugations $\rho\mapsto U\rho U^{-1}$
with $U$ either unitary or antiunitary. The subgroup of implementable transformations is the projective unitary group ${\rm PU}(N)$.
\end{lemma}
Note that ${\rm PUA}(N)$ can also be characterized as the set of transformations of the form $\rho\mapsto U\rho U^\dagger$ and $\rho\mapsto U\rho^T U^\dagger$,
with $U$ unitary. The transposition can be achieved by conjugation with an antiunitary map.

The proof of this lemma is simple. First,
we may assume that Alice and Bob have both agreed to represent quantum states by $N\times N$ unit trace positive semidefinite Hermitian matrices (that is, density matrices), which
thus constitute the set $\s$ of mathematical objects used as descriptions of the physical systems (cf.\ the notation in Section~\ref{SecReferenceFrames}).
Moreover, it is a physically
distinguished choice to encode quantum states such that \emph{statistical mixtures of quantum states are mapped to convex mixtures of the
corresponding density matrices}, as it is standard in quantum mechanics. The probabilistic interpretation of mixtures implies that different observers
cannot disagree on the question whether a given state is a statistical mixture of other given states.

Thus, every group element $T=\varphi_A\circ\varphi_B^{-1}$ must be a convex-linear symmetry of the set of $N\times N$ density matrices.
According to Wigner's Theorem~\cite{Bargmann},
this is equivalent to $T$ being a conjugation either by a unitary or antiunitary matrix.

Recalling the definition of implementability as explained in Figure~\ref{fig_scenario2}, we also see that only the subgroup of unitary conjugations is (physically)
implementable; antiunitary conjugations are not. They correspond to maps which are not completely positive.

The group ${\rm PUA}(N)$ is achievable, but is it minimal? The answer to this question clearly depends on the detailed physical background assumptions,
in particular on the question whether the physical carrier system (which Alice asks Bob to use) carries any physically distinguished choice of Hilbert space basis.
For example, if $N=2$, a carrier system could be given by a ground state and an excited state (say, of an atom) that span a two-dimensional Hilbert space.
In this case, Alice and Bob would under many circumstances agree on this orthonormal Hilbert space basis, and then $\g_{\min}=\{\mathbf{1}\}$ or
$\g_{\min}={\rm SO}(2)$, depending on whether Alice and Bob would also agree on the relative phase between the states.

Thus, the question of $\g_{\min}$ for transmitting quantum states only becomes interesting when we consider suitable additional physical background
assumptions. A crucial property of quantum systems in actual physics is that \emph{different systems can interact}. {In Subsection \ref{SubsecFullLab},
we will explore in detail the implications of this simple fact for the communication scenario of two observers in local quantum laboratories and, in particular, under which conditions an agreement on a qubit description can be employed to agree on the description of (almost) arbitrary $N$-level systems. But firstly, we briefly explain the agreement procedure for qubits.}

\subsection{Agreeing on qubit descriptions}\label{sec_qbit}

{Suppose Alice and Bob wish to agree on the description of an ensemble of qubits. For instance, they might wish to agree on the description of an ensemble of electron spins. Since a qubit density matrix $\rho$ has three degrees of freedom, this requires a tomographically complete set of three observables $\hat{M}_1,\hat{M}_2,\hat{M}_3$, for example the Pauli matrices $\sigma_x,\sigma_y,\sigma_z$.}

{Alice can ask Bob, through classical communication, to prepare a large number of copies of a fixed eigenstate $\rho_1$ with specified eigenvalue of $\hat{M}_1$ and to send these to her through the quantum communication channel (assumed to be noise-free and reversible, for simplicity). Upon receipt, Alice will perform tomography on the received systems in order to compare her description of the received states, call it $\rho_1'$, with her request; in general, she will find disagreement $\rho_1'\neq\rho_1$. The same task is repeated with copies of eigenstates $\rho_2,\rho_3$ of fixed eigenvalue of $\hat{M}_2,\hat{M}_3$, in general giving $\rho_2'\neq\rho_2$ and $\rho_3'\neq\rho_3$. As the three requested eigenstates $\rho_1,\rho_2,\rho_3$ constitute a basis for the qubit density matrices, this provides sufficient information for Alice to compute the unitary (or antiunitary) $U\in{\rm UA}(2)$, translating between the received and requested states as $\rho'_i=U\,\rho_i\,U^{-1}$, $i=1,2,3$, uniquely up to phase. If no preferred choice of Hilbert space basis exists, $U$ can, in principle, be any element of ${\rm UA}(2)$. By contrast, if a ``natural'' choice of basis exists, $U$ must be contained in a subset of ${\rm UA}(2)$, as argued above.}

The transformation $U$ defines a unique element of the achievable qubit transformation group $\g_{\rm qubit}\simeq{\rm PUA}(2)\simeq{\rm O}(3)$ which translates between Bob's and her descriptions of the qubits. This sets up an agreement on the description of arbitrary states, at least for the kind of qubits they employed in their procedure: if Alice henceforth wishes to receive a state $\rho_A$ (in her description), she will have to ask Bob, via classical communication, to prepare the state which takes the form $U^{-1}\,\rho_A\,U$ in his description.

As an aside, we note that the temporal stability of this qubit agreement procedure requires an additional assumption: if Bob sends two sets of copies of some fixed state at distinct times, then Alice will also receive them as two sets of copies of one state. As natural as this assumption may seem, it cannot hold under all physical circumstances; in particular, for accelerated observers in special relativity or even observers in a gravitational context, this condition may not be true in general. We shall call it therefore the {\it inertial frame condition} as it might be taken as an informational definition for two communicating frames to be inertial. We shall henceforth tacitly assume that it holds.

\subsection{Universal measurement devices}
\label{SubsecFullLab}
From a physical point of view, a natural question is whether we can describe the relation between the two laboratories of Alice and Bob,
including \emph{all} local quantum physics, by a single (small) group. Only in this case can we meaningfully take this group of transformations as constituting proper reference frame relations. If we consider Alice's local laboratory content as a collection of (many) finite-dimensional quantum
systems $S,S',S'',\ldots$, then Lemma~\ref{LemQuantumGeneral} only tells us that we can achieve ${\rm PUA}(N)$ for every system $S$ with $\dim S=N$; the tensor product over all these groups, each corresponding to one system, would form an achievable group. This is a huge group that seems highly inefficient for the communication task. If one could not do better, then Alice and Bob would have to negotiate a common reference frame for each of their local quantum systems separately. Clearly, there must be a way to do better than this. After all, there are relations between the different quantum systems that might be exploited in the communication task. For example, can we use the fact that different quantum systems interact?

Quantum systems carry quantum states, which are nothing but catalogues of probabilities of measurement outcomes. Therefore, we could relate different quantum systems $S$ and $S'$ to each other if we could somehow \emph{apply one and the same measurement device to both of them}. Let us elaborate on this idea, and imagine a device which we might call a  ``universal measurement device'', which is an apparatus that measures a given observable $\hat M$
for two quantum systems $S$ and $S'$ (of possibly different Hilbert space dimensions) universally.
In other words, we have a measurement device which accepts as inputs both systems $S$ and $S'$, and
in the former case measures an observable $\hat M(S)$, and in the latter another observable $\hat M(S')$. As a paradigmatic example, we can think of an idealized \emph{Stern-Gerlach
device} which measures, say, the spin in $z$-direction, $\hat S_z$, both for systems $S$ of spin $1/2$ and systems $S'$ of spin $1$ with the same given magnetic
field gradient, as sketched in Figure~\ref{fig_universal1}. Physically, the operators $\hat S_z(S)$
and $\hat S_z(S')$ represent the same quantity, but mathematically they correspond to observables on different Hilbert spaces.
\begin{figure*}[!hbt]
\begin{center}
\setlength{\fboxrule}{1pt}
\fbox{\includegraphics[angle=0, width=10cm]{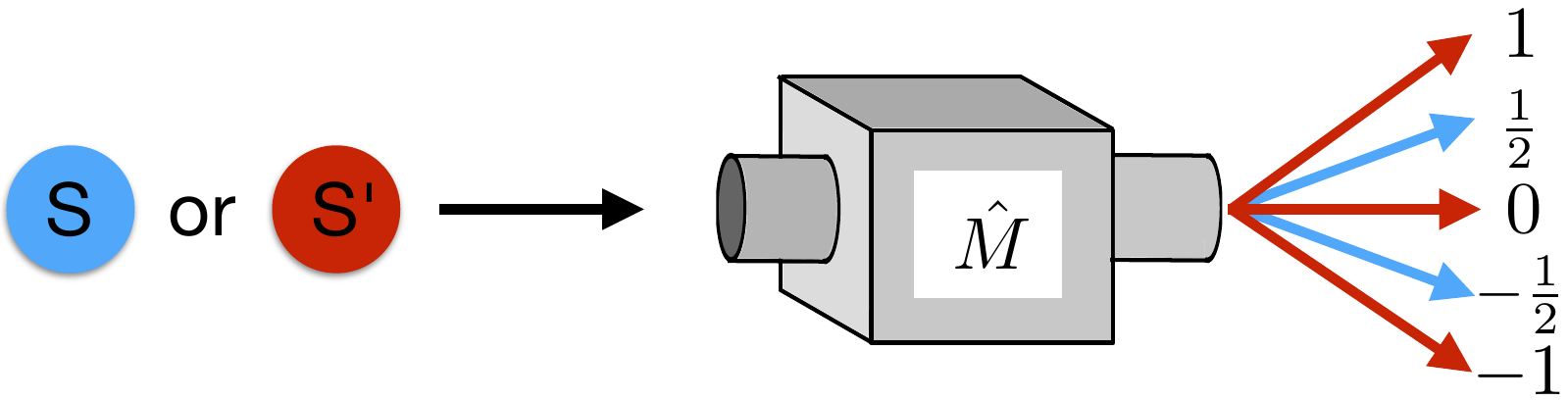}}
\caption{A Stern-Gerlach device as a universal measurement device. It measures the observable $\hat S_z$ (spin in $z$-direction) both on quantum systems $S$ of spin-$1/2$ and quantum systems $S'$ of spin-$1$, implementing mathematical observables $\hat S_z(S)$ and $\hat S_z(S')$ which are operators on $\C^2$ and $\C^3$ respectively.}
\label{fig_universal1}
\end{center}
\end{figure*}

Why should universal measurement devices exist at all, and what does it mean in general that $\hat M(S)$ and $\hat M(S')$ correspond to ``the same'' observable $\hat M$ on the
different Hilbert spaces of $S$ and $S'$? 
In fact, these two questions are intimately intertwined. Namely, a minimal prerequisite for operationally defining ``that very same quantity" $\hat M$ for physically distinct systems $S,S'$ is that there exists, in principle, a measurement device in which $\hat M$ can be measured for both $S$ and $S'$ -- ideally in the same setting.\footnote{For instance, returning to the Stern-Gerlach device as a universal measurement device for spin, the setting would be determined by the direction of the magnetic field gradient.} Otherwise, there does not exist an unambiguous way of comparing $\hat{M}(S)$ and $\hat M(S')$ and, subsequently, to evaluate them as being ``the same quantity'', but carried by physically distinct systems. In other words, if $\hat M$ is supposed to be a quantity which can be carried by physically distinct systems, it better be universally measurable and thus universal measurement devices better exist. 

A universally measurable quantity which can be defined for a large number of different kinds of systems can therefore also be regarded  as a ``carrier independent" quantity. One might then be inclined to ask why carrier independent quantities should exist at all. If such carrier independent quantities did not exist, it would be difficult to operationally define the notion of interaction between physically distinct systems. By an interaction, one usually thinks of a process where physical systems exchange or redistribute some quantities or properties. One tacitly follows the intuition that a quantity that is being transferred from one system to another is in its nature somehow ``the same'' before and after the exchange. But if the carrier systems are physically distinct then (at least the nature of) the quantity should be carrier independent. In particular, a carrier independent $\hat M$ might even be a ``conserved quantity''; that is, a physical quantity whose
total value is typically preserved in closed systems, even in interactions between different kinds of quantum systems, which in the end allows to compare the value
of $\hat M$ on $S$ and $S'$.

Interactions, in turn, are at the heart of any measurement. Hence, if carrier independent quantities did not exist, it would be troublesome to envisage how interesting measurement devices could be built in the first place which do not employ only one kind of system to also measure only the same kind of system. Instead, with carrier independent quantities available, one can imagine that quantum systems of some kind can be used to build devices in which they interact in ``canonical ways'' with -- and thereby measure --
systems of another kind. For example, we can use the spins of many electrons (which are qubits) to build a magnet which in turn can be used in a Stern-Gerlach device,
defining a quantization axis for particles of higher spin.

These arguments suggest that an interesting (`interactive') physical world renders the existence of such universally measurable observables and corresponding universal measurement devices quite natural. In the following, we shall thus simply \emph{assume} their existence, and our aim is to investigate the consequences of this assumption, in particular, to derive the symmetry group that is implied by their properties. This reverses the prevalent logic in the standard literature of presupposing a spacetime symmetry group which, in turn, implies conserved quantities (such as energy, momentum, angular momentum etc.) that, in light of our discussion, would be carrier independent and universally measurable.

In order to analyze the implications of universal measurement devices in our communication scenario more formally, we have to specify the mathematical assumptions that we make about their behavior. It turns out that we do not need to specify too many details about their working for the present purpose; only the following quite intuitive property is needed:
\begin{assumptions}
\label{AssUniversal3}
When we say that (one, or several) observables $\hat M$ can be \emph{universally measured on quantum systems $S$ and $S'$}, we assume that there are operators $\hat M(S)$ and $\hat M(S')$ such that $\hat M(S)$ uniquely determines $\hat M(S')$, and vice versa. Moreover, we assume that this interdependence is continuous in both directions\footnote{An example
is given by $\hat M=\hat S_z$, if $S$ and $S'$ are spin-$1/2$ and spin-$1$ systems, see also Example~\ref{ExEncoding} below. In particular, a continuous change of $\vec{n}\in\mathbb{R}^3$ in the observable family $\hat{S}_{\vec{n}}=\vec{n}\cdot\vec{S}$ for $\vec{S}=(\hat{S}_x,\hat{S}_y,\hat{S}_z)$ for spin-$1/2$ induces a continuous change of the same family for spin-$1$, and vice versa.
However, if $S$ is a spin-$0$ system and $S'$ a spin-$1/2$ system, then
we do \emph{not} have this property. For all unit vectors $\vec n\in\mathbb{R}^3$, we would have $\hat S_{\vec n}(S)=0$, and knowledge of this observable
would not allow to infer $\hat S_{\vec n}(S')$.}.
\end{assumptions} 
These assumptions are indeed quite intuitive: small changes of an observable $\hat M$ on a system $S$ should lead to small changes of that observable on other quantum systems $S'$. For this to make sense, there must be an invertible map $\hat M(S)\mapsto \hat M(S')$ which encodes what we mean by ``the same measurement''.

In Section~\ref{sec_phys}, we will consider a more general scenario (where the eigenvalues are themselves physical quantities, not only abstract labels as in this section). There we will have to reconsider the properties of universal measurement devices and extend the set of assumptions.

Let us now return to our communication scenario, and ask the question how Alice and Bob can use universal measurement devices to simplify their task. If there are ``enough'' observables $\hat M$ that can be universally measured on quantum systems $S$ and $S'$, then it seems intuitively clear that the two agents can exploit this fact: by agreeing on a description of states on $S$, they should be able to obtain a common description of states on $S'$ ``for free''. The following lemma specifies the conditions under which this is possible.

\begin{lemma}
\label{Lem3.3}
Suppose that a set of observables $\mathcal{M}:=\{\hat M_i\}_{i\in I}$ can be universally measured on two quantum systems $S$ and $S'$,
and suppose this set is large enough
to be \emph{tomographically complete} on $S'$.\footnote{That is, the set of outcome probabilities
${\rm tr}(\rho'\, \pi_i^{(j)})$, with $\pi_i^{(j)}$ the eigenprojectors of $\hat M_i(S')$,
determines the physical state $\rho'\in\s_{\rm phys}(S')$ uniquely. Note that this does not necessarily mean that the expectation values ${\rm tr}(\rho \hat M_i(S'))$ determine the
state $\rho'$ uniquely. As a counterexample, consider the spin-$1$ matrices $S_{\vec n}:=n_x S_x + n_y S_y + n_z S_z$ with $\vec n \in\R^3$ the unit vectors.
These matrices span a three-dimensional linear subspace of the space of observables (a representation of the Lie algebra $\mathfrak{so}(3)$);
thus, the set of numbers ${\rm tr}(\rho S_{\vec n})$ reveals only three of
the eight independent parameters of $\rho'$. However, since we have an irreducible representation, knowing all outcome probabilities on the eigenvectors is sufficient
to determine $\rho'$.} Then there is a protocol that allows Alice and Bob to agree on an encoding $\varphi^{(S')}:\s_{\rm phys}(S')\to\s(S')$ by exchanging only classical information, given that they already agree on an encoding $\varphi^{(S)}:\s_{\rm phys}(S)\to \s(S)$.
Moreover, if \emph{all} observables of $S$ can be universally measured on $S$ and $S'$, then the resulting map $\varphi^{(S)}\mapsto\varphi^{(S')}$ is continuous. That is, continuously changing $\varphi^{(S)}\to\varphi^{(S)}+\delta\varphi^{(S)}$, while keeping the exchange of classical information identical, changes $\varphi^{(S')}$ continuously.

In the case where $\mathcal{M}(S)$ is a strict subset of all observables on $S$ (i.e.\ not all observables can be universally measured), we have a slightly weaker continuity property. To state it, call two encodings $\varphi^{(S)}$ and $\tilde\varphi^{(S)}$ \emph{equivalent} if they map the physical observables\footnote{We are slightly abusing notation,
by writing $\varphi^{(S)}(\hat M(S))$ for the corresponding encoding of an \emph{observable} $\hat M(S)$, even though we have defined $\varphi^{(S)}$
to act on states only. However, this notation makes sense: if we treat $\rho_{\rm phys}$ and $\hat M(S)$ as unknown matrices,
then $\varphi^{(S)}(\rho_{\rm phys})=U\rho_{\rm phys} U^\dagger$ for some unknown unitary $U$ and possibly an additional transposition on $\rho_{\rm phys}$.
Consequently, we obtain $\tr(\rho_{\rm phys}\hat M(S))=\tr\left(\varphi^{(S)}(\rho_{\rm phys})\varphi^{(S)}(\hat M(S))\right)$,
and we can consistently claim that $\varphi^{(S)}$ is also the correct encoding of observables into matrices.} $\mathcal{M}(S)$ to the same set of mathematical operators, i.e.\ if $\varphi^{(S)}\left(\mathcal{M}(S)\right)=\tilde\varphi^{(S)}\left(\mathcal{M}(S)\right)$. Then continuously changing $\varphi^{(S)}$ to an equivalent encoding $\varphi^{(S)}+\delta \varphi^{(S)}$ continuously changes $\varphi^{(S')}$.
\end{lemma}
The method to obtain $\varphi^{(S')}$ from $\varphi^{(S)}$ again is quite intuitive, given the communication scenario of Figure~\ref{fig_scenario1}. Suppose that we have a situation as described in Lemma~\ref{Lem3.3},
and Alice and Bob have agreed on a common encoding $\varphi^{(S)}$ of quantum system $S$. In order to agree
with Bob on an arbitrary encoding $\varphi^{(S')}$ of $S'$, Alice can do the following.
Suppose for concreteness that $I=\{1,2,3\}$. Then Alice can say:
\emph{``Bob, please build the three universal measurement devices that measure the observables $\hat M_1(S), \hat M_2(S), \hat M_3(S)$ on $S$. As you know, you can also send quantum systems $S'$ into those devices. From now on, let's use the following matrices as descriptions of the observables $\hat M_1(S'), \hat M_2(S'), \hat M_3(S')$: [tables of numbers]. Thus, you know how to do tomography with quantum states on $S'$, and we can agree over the phone on their description.''}

In this request, the $\hat M_i(S)$ can be described by sending the matrices
$\hat M_i(S)_0:=\varphi^{(S)}\left(\hat M_i(S)\right)$. The resulting encoding $\varphi^{(S')}$ will then be shared between Alice and Bob. It is not unique,
but depends on the choice of matrix encodings that Alice suggests. These constitute classical data that can be communicated between Alice and Bob.
Note however that this protocol fails if the $\hat M_i(S)_0$ do not encode universal observables; in other words, the protocol, resp.\ algorithm, is tailored to the set of matrices $\varphi^{(S)}(\mathcal{M}(S))$. The formal details are as follows.
\begin{proof}
Let $\varphi_A^{(S')}$ be Alice's personal choice of classical description of quantum states on $S'$, and let $T_{S\to S'}^{\rm phys}$ be the map which satisfies $T_{S\to S'}^{\rm phys}\hat M_i(S)=\hat M_i(S')$. According to Assumptions~\ref{AssUniversal3}, it is well-defined and a homeomorphism into its image. Define
\[
   T_{S\to S'}:=\varphi_A^{(S')}\circ T_{S\to S'}^{\rm phys}\circ\left(\varphi^{(S)}\right)^{-1},
\]
which is thus also a homeomorphism into its image. It satisfies $T_{S\to S'}\hat M_i(S)_0=\hat M_i(S')_0:=\varphi_A^{(S')}(\hat M_i(S'))$. Let $J\subset I$ be a finite set such that the observables $\{\hat M_j(S')\}_{j\in J}$ are still tomographically complete on $S'$ (if $I$ itself is finite, we may have $J=I$). Alice can communicate the matrix descriptions $\{\hat M_j(S')_0\}_{j\in J}$ and $\{\hat M_j(S)_0\}_{j\in J}$ to Bob. Using the latter, Bob can locally implement the physical measurement devices $\{\hat M_j(S)\}_{j\in J}$ (assuming agreement on $\varphi^{(S)}$), which by universality amounts to an implementation of the observables $\{\hat M_j(S')\}_{j\in J}$. He therefore knows the $\varphi_A^{(S')}$-description of a tomographically complete set of observables (and thus of their eigenprojectors which span the space of Hermitian operators), which allows him to deduce the $\varphi_A^{(S')}$-description of all physical quantum states on $S'$.

Our description of the continuity property seems puzzling at first sight: if $\varphi_A^{(S')}$ is simply Alice's personal choice of encoding of $S'$, then why should this change if we change $\varphi^{(S)}$? Can she not make her choice of $\varphi_A^{(S')}$ independently of the agreed-upon $\varphi^{(S)}$? Of course she can; but changing $\varphi_A^{(S')}$ while keeping $\varphi^{(S)}$ changes $T_{S\to S'}$ and thus changes the classical information that is sent from Alice to Bob. To have a \emph{fixed} exchange of classical information, however, $\varphi_A^{(S')}$ must change with $\varphi^{(S)}$, and we can define a map
\[
   \varphi^{(S)}\to\varphi_A^{(S')}\equiv \varphi^{(S')}.
\]
Note that
\begin{equation}
   \varphi^{(S')}=T_{S\to S'}\circ \varphi^{(S)}\circ T_{S'\to S}^{\rm phys}\,\,,
   \label{eqDefSSPrime}
\end{equation}
and for fixed exchange of classical information, $T_{S\to S'}$ is fixed, and so is $T_{S'\to S}^{\rm phys}$ which, as a physical map, is independent of any choice of description. In this equation, consider replacing $\varphi^{(S)}$ by $\tilde\varphi^{(S)}:=\varphi^{(S)}+\delta\varphi^{(S)}$. The resulting expression will still be well-defined if
and only if $\tilde\varphi^{(S)}\left(\mathcal{M}(S)\right)$ agrees with the domain of definition of $T_{S\to S'}$, which is $\varphi^{(S)}\left(\mathcal{M}(S)\right)$.
If this is the case, we will have a continuous change of $\varphi^{(S')}$.
\end{proof}

We will look in more detail into this protocol for spin systems in our concrete physical world in Example~\ref{ExEncoding} below. The next step will be to see how Alice and Bob can relate their full local laboratories by using the protocol above.

\subsection{Relating local laboratories: the universal measurability graph}

Now we formalize the idea of the beginning of the previous subsection: if there is a single quantum system that ``interacts naturally'' indirectly with {\it all} other systems,
then a choice of reference frame for that system implies a choice of reference frame for the full laboratory.

\begin{definition}[Universal measurability graph]
\label{DefGraph}
Consider the set of all finite-dimensional quantum systems $S,S',S'',\ldots$; we regard them as vertices of a graph.
Draw a directed edge from $S$ to $S'$ if and only if the situation of Lemma~\ref{Lem3.3} holds, i.e.\ if there is a set of observables $\{\hat M_i\}_{i\in I}$ which is tomographically complete on $S'$, and which is universally measurable on $S$ and $S'$.
\end{definition}

In Subsection~\ref{sec_concrete}, we will look at the concrete realization of the universal measurability graph in our own universe (in particular, see Figure~\ref{fig_graph}).

\begin{assumptions}
\label{AssInteraction}
There exists at least one quantum system $S$ that is a ``root'' of this graph, in the sense that every vertex can be reached
from $S$ by following directed edges. Furthermore, we assume that no quantum system with a partially preferred choice of basis or encoding
is a root of this graph.
\end{assumptions}

Choose one root which has the smallest Hilbert space dimension $d$ among all roots, and call it $S$. Clearly, $d=\dim\, S\geq 2$.

Suppose that $S'$ is any quantum system such that there is a directed edge directly from the root $S$ to $S'$. By assumption, $S$ does \emph{not}
carry any natural choice of Hilbert space basis, in the sense that there is no distinguished subset of encodings $\varphi^{(S)}$ at all.
We know how to parametrize all possible
encodings $\varphi^{(S)}$: given an arbitrary fixed encoding $\varphi$, all others can be written in the form $\varphi^{(S)}(\rho)=
T(\varphi(\rho))$ for all $\rho\in\s_{\rm phys}$, where $T\in {\rm PUA}(d)$. In this case, we write $\varphi^{(S)}=\varphi^{(S)}_T$.
Due to Lemma~\ref{Lem3.3}, the directed edge induces a corresponding set of encodings $\varphi^{(S')}_T$ on $S'$, again labelled by $T\in{\rm PUA}(d)$.
For every $T,V,W\in{\rm PUA}(d)$, we have
\[
   \varphi_W^{(S)}\circ\left(\varphi_V^{(S)}\right)^{-1}\circ \varphi_T^{(S)}=\varphi_{WV^{-1}T}^{(S)}
\]
for the root $S$.
If $S'$ is a system such that all observables of $S$ are universally measurable on $S$ and $S'$, we get due to~(\ref{eqDefSSPrime}) (dropping some ``$\circ$'')
\begin{eqnarray}
   \varphi_W^{(S')}\circ \left(\varphi_V^{(S')}\right)^{-1}\circ \varphi_T^{(S')}&=&
   T_{S\to S'}\varphi_W^{(S)} T_{S'\to S}^{\rm phys}\left(T_{S\to S'}\varphi_V^{(S)}  T_{S'\to S}^{\rm phys}\right)^{-1}
   T_{S\to S'} \varphi_T^{(S)}  T_{S'\to S}^{\rm phys}\nonumber\\
   &=& T_{S\to S'}\circ \varphi_{W V^{-1} T}^{(S)}\circ T_{S'\to S}^{\rm phys}=\varphi_{WV^{-1} T}^{(S')}.
   \label{eqTripleProd}
\end{eqnarray}
Define
\[
   G_W^{(S')}:=\varphi_W^{(S')}\circ\left(\varphi_{\mathbf{1}}^{(S')}\right)^{-1},
\]
which is a linear operator on $S'$. Then we get for all $V,W\in{\rm PUA}(d)$
\begin{equation}
   G_W^{(S')} G_V^{(S')} = \varphi_W^{(S')}\circ \left(\varphi_{\mathbf{1}}^{(S')}\right)^{-1}\circ \varphi_V^{(S')}\circ\left(\varphi_{\mathbf{1}}^{(S')}\right)^{-1}
   =\varphi_{WV}^{(S')}\circ\left(\varphi_{\mathbf{1}}^{(S')}\right)^{-1} = G_{WV}^{(S')}.
   \label{eqGroupRep}
\end{equation}
Since the map $W\mapsto G_W^{(S')}$ is continuous, it is a group representation of ${\rm PUA}(d)$ within ${\rm PUA}(S')$.
By considering the connected component at the identity, we thus get a projective representation of ${\rm PU}(d)$ on\footnote{More precisely, $G_W^{(S')}$ acts on the density matrices of $S'$ by conjugation, which is a proper representation of ${\rm PU}(d)$ on the vector space of Hermitian matrices. The corresponding representation on the Hilbert space of $S'$ is however in general only a projective representation.} $S'$.
The same conclusion holds for quantum systems $S'$
that are not directly connected to $S$, but can be reached from $S$ by a path of several directed edges in the graph.\footnote{If there are several
paths leading from the qubit $S$ to $S'$, then a choice has to be made as to which path to take to assign a resulting encoding map $\varphi_T^{(S')}$.
However, this is also a choice that can be communicated by classical information between observers. If $S\to S''\to S'$ in the universal measurability graph,
then universal measurability will be transitive, in the sense that if the set of all $S$-observables $\mathcal{M}$ is universally measurable on $S$ and $S''$, and also on $S''$ and $S'$, then this set is also universally measurable on $S$ and $S'$. However, this set need not be tomographically complete on $S'$, which is why there need not be an edge in the universal measurability graph going from $S$ to $S'$ directly.}

Now suppose that $S'$ is a system such that the set $\mathcal{M}(S)$ of observables that are universally measurable on $S$ and $S'$ is a strict \emph{subset} of
all physical observables of the root $S$. According to~(\ref{eqDefSSPrime}), the calculation in~(\ref{eqTripleProd}) continues to make sense
if $\varphi_T^{(S)}(\mathcal{M}(S))=\varphi_V^{(S)}(\mathcal{M}(S))=\varphi_W^{(S)}(\mathcal{M}(S))$, and thus~(\ref{eqGroupRep}) remains
true if $\varphi_V^{(S)}(\mathcal{M}(S))=\varphi(\mathcal{M}(S))$ and $\varphi_W^{(S)}(\mathcal{M}(S))=\varphi(\mathcal{M}(S))$. Since $\varphi_V^{(S)}=V\circ\varphi$
and similarly for $W$, this means that $V$ and $W$ are in the subgroup of transformations of ${\rm PUA}(d)$ that preserve the image of the observables that are universally measurable on $S$ and $S'$. In other words, we obtain a projective representation of this subgroup.

Let $\mathcal{C}$ be the set of all systems $S'$ such that all observables of $S$ are universally measurable on $S$ and $S'$. (Clearly, $\cc$ contains $S$.)
The formal product
\begin{equation}
   \Phi:=\left\{\varphi_T\,\,|\,\, T\in {\rm PUA(d)}\right\},\qquad \mbox{where }\varphi_T:=\bigotimes_{S'\in\mathcal{C}} \varphi_T^{(S')}
   \label{eqAllFactors}
\end{equation}
defines a physically distinguished set of encodings of \emph{all} these systems at once.
Suppose there was a smaller subset of distinguished encodings
$\Phi'\subsetneq\Phi$, then $\Phi'=\{\varphi_T\,\,|\,\, T\in\g\subsetneq {\rm PUA}(d)\}$, with $\g$ a strict subset of ${\rm PUA}(d)$.
Operationally, being able to encode all laboratory quantum systems via some $\varphi_T$ implies in particular that one can encode the root $S$ via $\varphi_T^{(S)}$; thus, the set $\{\varphi_T^{(S)}\,\,|\,\, T\in\g\}$ would constitute a subset of distinguished encodings of the root $S$, which contradicts our assumption that there is no such subset. Hence $\Phi$ must be a minimal set of encodings.

But then, in particular~(\ref{eqGroup}) must hold, such that for every choice
of $T,V,W\in {\rm PUA}(d)$, there is some $X\in{\rm PUA}(d)$ such that $\varphi_W\circ \varphi_V^{-1}\circ\varphi_T =\varphi_X$. Since
\[
   \varphi_W^{(S)}\circ\left(\varphi_V^{(S)}\right)^{-1}\circ \varphi_T^{(S)}=\varphi_{WV^{-1}T}^{(S)}
\]
for the root $S$, it follows that $X=WV^{-1}T$. This equation must be consistent with the other factors of $\varphi_T$
in accordance with~(\ref{eqAllFactors}), which is only possible if an analogous equation holds for all systems $S'\in\mathcal{C}$ simultaneously,
yielding an independent proof of~(\ref{eqTripleProd}). Most importantly, the group associated to $\Phi$, that is
$\mathcal{G}:=\{\varphi_1\circ \varphi_2^{-1}\,\,|\,\, \varphi_1,\varphi_2\in\Phi\}$ is isomorphic to ${\rm PUA}(d)$. This is the minimal group that translates between Alice's and Bob's encodings, if they describe the totality of all systems $S'\in\mathcal{C}$ in their laboratories.

It is tempting to generalize~(\ref{eqAllFactors}), and to define the product to range over \emph{all} finite-dimensional quantum systems $\mathcal{\bar C}$,
including those $S'$ for which not all  observables of $S$ are universally measurable on $S$ and $S'$. This does not cause any problems for systems $S'$ which carry \emph{no} universally measurable observable at all (or only trivial such observables $\lambda\cdot\mathbf{1}$ with $\lambda\in\R$).
The conclusions will be the same as above, namely that these systems
carry projective representations of ${\rm PU}(d)$ (which will typically be trivial representations), and the set of encoding
$\bigotimes_{S'} \varphi_T^{(S')}$, where the product is over all those $S'$ \emph{and} all $S'\in\mathcal{C}$, will be a minimal set of encodings.

However, a subtle difficulty arises with quantum systems $S'$ that carry a non-trivial strict subset $\mathcal{M}$ of universally measurable observables (among those observables of $S$). As a concrete physical example, suppose $S$ is a spin-$1/2$ particle while $S'$ is a photon moving in $z$-direction. Then it is meaningful to ask for the photon spin in $z$-direction, however, not for the spin in any other direction. Hence, the spin component $\hat{S}_z$ and $\hat{S}_z^2\propto\mathbb{1}$ are universally measurable on $S$ and $S'$ in this case -- in contrast to $\hat{S}_x,\hat{S}_y$ (see also Section \ref{sec_concrete} below).

Our picture of coding and encoding of quantum states on $S'$ rests on a tacit assumption: namely that the agents (Alice and Bob in the communication scenario) have perfect knowledge on the \emph{choice of quantum system $S'$}, which in this case \emph{includes} the specification of the set of observables $\mathcal{M}$ that are universally measurable on $S$ and $S'$. For instance, in the case of the spin-$1/2$ particle and photon propagating in $z$-direction, this would be tantamount to Alice and Bob already agreeing on the $z$-direction beforehand. In general, this knowledge would help Alice and Bob to encode quantum states of the root $S$, inducing a strict subset of physically distinguished
encodings $\varphi^{(S)}$ among all encodings $\varphi^{(S)}_T$ with $T\in{\rm PUA}(d)$. {In the photon example, it would be natural to choose the eigenstates of $\hat{S}_z$ as Hilbert space basis which would only leave relative phases to determine.} Thus, taking the formal product in~(\ref{eqAllFactors})
over \emph{all} quantum systems $S'\in\mathcal{\bar C}$ would yield a set of encodings which is not minimal.
Not only would this contradict our assumptions,
but it would also ignore the {operational} difficulty of setting up the agreement of the choice of $S'$ between Alice and Bob. For example, in the case of the spin-$1/2$ particle and photon propagating in $z$-direction it would ignore the problem of firstly having to agree on the $z$-direction.

In order to deal with this situation, we have to treat quantum systems $S'$ with non-trivial subsets of universally measurable observables $\mathcal{M}$ differently. Suppose that $S'$ and $\tilde S'$ are physically identical \emph{except} for the corresponding sets of universally measurable observables $\mathcal{M}$ and
$\mathcal{\widetilde M}$, and that those sets are related by unitary conjugation and/or transposition on $S$, i.e.\ that there is $T\in{\rm PUA}(d)$ with
$\mathcal{\widetilde M}(S)=T\left(\mathcal{M}(S)\right)$. For instance, $S'$ could be a photon propagating in $z$-direction, while $\tilde S'$ could be a photon propagating in $x$-direction. In this case, $\hat{S}_z$ and $\hat{S}_x$ are indeed related by a unitary.
Then treat them as two instances\footnote{We can classify systems $S'$ also with respect to the question which observables are $(S'',S')$-co-measurable
with respect to \emph{other} systems $S''$. However, here we are only interested in the question how we can use the agreement on the encoding of the root $S$
to agree on the encoding of $S'$, which is why only the case $S''=S$ is interesting here.} of the \emph{same} system {(or equivalence class of systems)} $[S']$, which is however
in two different states $(\mathcal{M}(S),\rho_{S'})$ resp.\ $(\mathcal{\widetilde M}(S), \rho_{\tilde S'})$. That is, we consider the set of universally measurable observables
as part of the specification of the state. {In the photon case, $[S']$ would just mean: a photon -- without specification of propagation direction. This is something Alice and Bob can agree on by classical communication.}

This prevents the difficulty just mentioned. For every $T\in {\rm PUA}(d)$, we can define a corresponding encoding of $[S']$ via the map $(\mathcal{M}(S),\rho_{\rm phys})\mapsto (\varphi_T^{(S)}(\mathcal{M}(S)), \varphi_T^{(S')}(\rho_{\rm phys}))$, and define the analog of~(\ref{eqAllFactors}) by taking the product over
all $[S']$, obtaining a set $\bar\Phi$ of encodings $\varphi_T$ with $T\in{\rm PUA}(d)$ of the totality of all quantum systems. Simply agreeing on some $[S']$ (i.e., in the photon case simply agreeing on $[S']$ being a photon, but not necessarily on its propagation direction) does not allow to simplify agreement on the root $S$, and $\bar\Phi$ will be a minimal set of encodings.
Arguing as above, we get $\mathcal{G}_{\min}={\rm PUA}(d)$ as abstract minimal group. Summarizing our findings, and reconsidering the assumptions that went into their derivation, we arrive at the following result.

\begin{theorem}
\label{TheQubitNormed}
Consider the communication scenario of Figure~\ref{fig_scenario1}: Alice and Bob attempt to agree on a classical description of finite-dimensional quantum systems; in fact, a description that works for the totality of \emph{all} quantum systems in their labs. Moreover, in their world, there are many ``universally measurable'' observables in the sense of Assumptions~\ref{AssUniversal3}: the \emph{universal measurability graph} (cf.\ Definition~\ref{DefGraph}) has roots (Assumptions~\ref{AssInteraction}), and the smallest Hilbert space dimension of any root is $d$.
Then, the minimal group that translates between their descriptions (in the sense of Definition~\ref{DefMinimal}) is
\[
   \g_{\min}={\rm PUA}(d).
\]
In particular, if we assume that there exists a ``qubit root'' $S$ of dimension $\dim\,S=2$, then
\[
   \g_{\min}={\rm O}(3),
\]
and the subgroup of implementable transformations (cf.\ Figure~\ref{fig_scenario2}) is ${\rm SO}(3)$. Moreover, in this case, if $S'$ is any other quantum system such that all observables of $S$ are universally measurable on it, then $S'$ carries a projective representation of ${\rm SO}(3)$. All other quantum systems $S'$ carry a projective representation of the subgroup of ${\rm SO}(3)$ which leaves the set of universally measurable observables invariant.
\end{theorem}

Thus, our assumptions have reconstructed an important property of quantum theory in our universe: that many systems come with a representation of the rotation group
in three dimensions. The usual point of view is that this is a \emph{consequence} of the three dimensions of space. However, here we argue the other way around:
the fundamental theory is quantum mechanics, and the emergence of an ${\rm SO}(3)$-symmetry can be understood on this basis alone. This insight is
very much in spirit of von Weizs\"acker's ``ur theory''~\cite{Weizsaecker}, but goes far beyond it by putting the argumentation on firm operational grounds: as long as there is a ``seed qubit'' (or root qubit) $S$ that carries enough observables which can be jointly (or universally) measured on other quantum systems, we obtain the symmetry group of three-dimensional rotations. In this sense, one can indeed view qubits as the fundamental building blocks of Alice's and Bob's world, given that the full reference frame transformation group ${\rm SO}(3)$ follows from the possible relations between their respective qubit descriptions alone.

In the next subsection, we will explore in more detail how our abstract assumptions (and Theorem~\ref{TheQubitNormed}) are
concretely realized in our actual physical world.

\subsection{Concrete realization in our universe}\label{sec_concrete}
It is important to notice that our derivations so far were purely abstract, without any assumptions of an underlying spacetime structure. We have arrived at the symmetry group ${\rm SO}(3)$ without assuming the dimensionality of space, either Galilean or special relativity, or other concrete properties of spacetime as observed in our world. In this section, we have a look at our actual universe in the context of relativistic quantum mechanics, and see how the abstract notions and assumptions from above (universal measurability graph etc.) are concretely physically realized there.

Let us consider three kinds of finite-dimensional quantum systems that exist in our universe:
\begin{itemize}
\item A spin-$1/2$ qubit encoded into an electron spin (call this quantum system $S$);
\item a spin-$1$ degree of freedom $S'$ -- either elementary (as in a $W$ or $Z$ boson), or as an effective degree of freedom
(such as the nuclear spin in orthohydrogen);
\item photon polarization qubits $S''$.
\end{itemize}

\begin{figure*}[!hbt]
\begin{center}
\setlength{\fboxrule}{1pt}
\fbox{\includegraphics[angle=0, width=10cm]{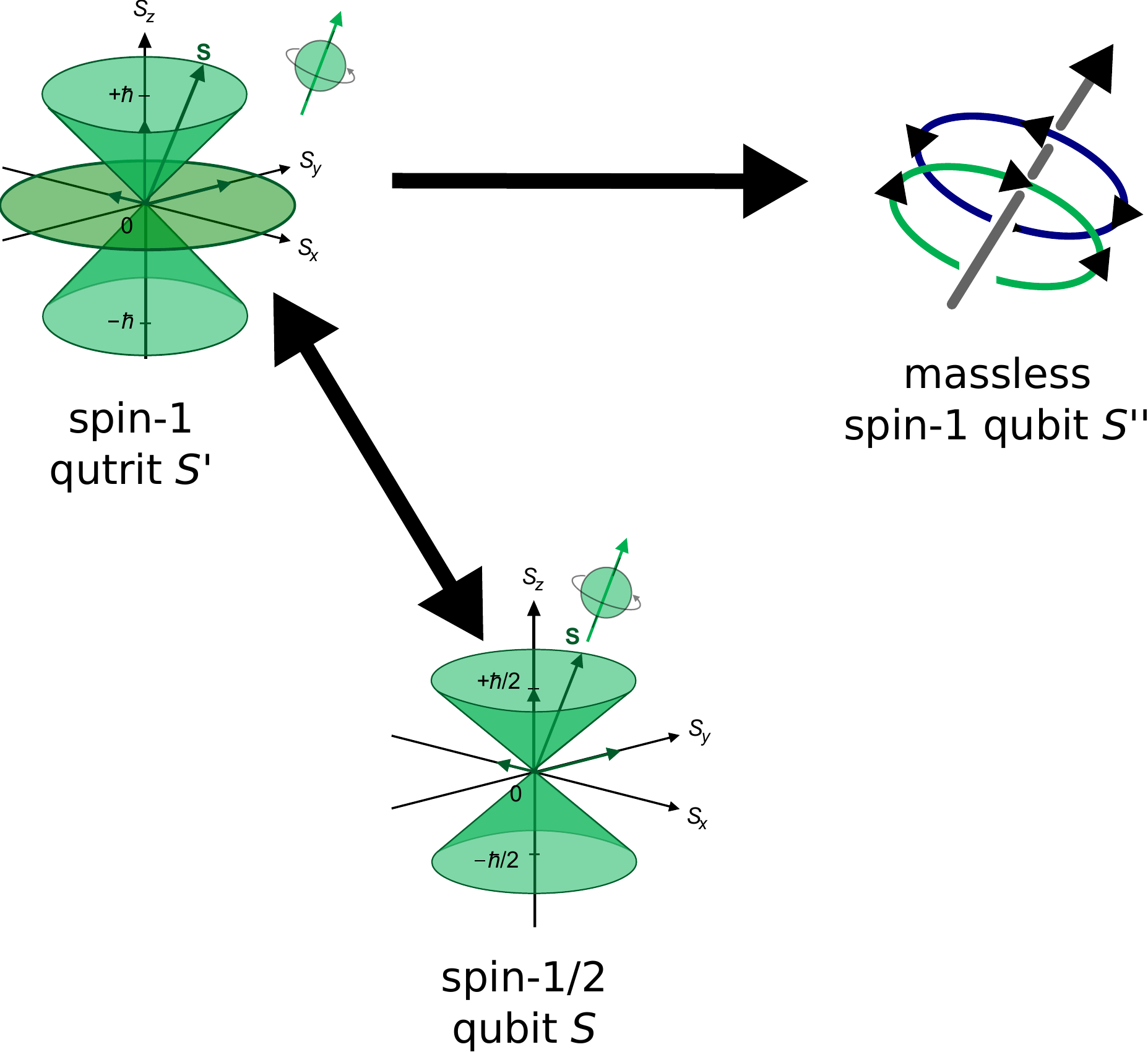}}
\caption[Caption]{Part of the universal measurability graph in relativistic quantum mechanics. Massive spin-$1/2$ particles $S$ (like an electron)
and massive spin-$1$ particles $S'$ have a set of universally measurable observables which is tomographically complete on both, namely spin in all directions.
For massive spin-$1$ particles $S'$ and photon polarization $S''$, however, only the set of observables which is supported in the span of $|+1\rangle$
and $|-1\rangle$ of spin in the direction of the photon momentum is universally measurable, which is tomographically complete on $S''$, but not on $S'$.
Thus, we can indirectly lift an encoding map $\varphi^{(S)}$ of the electron spin qubit to an encoding map $\varphi^{(S'')}$ of the photon
polarization qubit, but not vice versa.\footnotemark}
\label{fig_graph}
\end{center}
\end{figure*}

Figure~\ref{fig_graph} shows part of the universal measurability graph as defined in Definition~\ref{DefGraph}, namely the vertices $S$, $S'$ and $S''$.
Arrows are drawn according to universal measurability. In particular, there is an arrow from $S$ to $S'$,
and another arrow from $S'$ to $S$, since both systems can be jointly measured in Stern-Gerlach devices. More generally, let $I$ be the set
of all unit vectors in $\R^3$, then the set of spin observables $\{S_{\vec n}\}_{\vec n \in I}$ is universally measurable on $S$ and $S'$, namely by constructing a Stern-Gerlach device with magnetic field (and inhomogeneity) in direction $\vec n$.

\footnotetext{Picture source for massive spins: Wikipedia (declared as Public Domain).}

We will now give a detailed exposition of how the protocol in Lemma~\ref{Lem3.3} works in this case. While the following example is physically not particularly surprising, it nevertheless illustrates the protocol and the claimed continuous dependence of $\varphi^{(S')}$ from $\varphi^{(S)}$.

\begin{example}[Encoding: from spin-$1/2$ to spin-$1$]
\label{ExEncoding}
According to~\cite{Hofmann}, there are five spatial directions $\vec n_1, \ldots \vec n_5\in\mathbb{R}^3$ with $|\vec n_i|=1$ such that
$\{\hat S_{\vec n_i}\}_{i=1,\ldots,5}$ is tomographically complete on the spin-$1$ state space. Concretely, we can choose
\[
   \vec n_1=(1,0,0)^T,\quad \vec n_2=(0,1,0)^T,\quad \vec n_3=\frac 1 {\sqrt{2}}(1,1,0)^T,\quad
   \vec n_4=\frac 1 {\sqrt{2}}(0,1,1)^T,\quad\vec n_5=\frac 1 {\sqrt{2}}(1,0,1)^T.
\]
The actual data to start with is the set of physical spin-$1/2$ observables $\hat S_{\vec n_i}(S)$, represented mathematically as $\hat S_{\vec n_i}(S)_0=\vec n_i\cdot\vec\sigma$ via an encoding map $\varphi^{(S)}$,
with $\vec\sigma=(\sigma_x,\sigma_y,\sigma_z)$ the Pauli matrices. By using universal Stern-Gerlach devices, this also defines the physical observables
$\hat S_{\vec n_i}(S')$. Now we have to construct an arbitrary matrix representation $\hat S_{\vec n_i}(S')_0$ of these observables.
All possible choices are related by unitary conjugation
and possibly transposition; we can choose one arbitrarily. To this end, for every unit vector $\vec n\in\R^3$, set $\hat S_{\vec n}(S')_0:=\vec n \cdot\vec S$,
where $\vec S:=(\hat S_x,\hat S_y,\hat S_z)$, and
\begin{equation}
   \hat S_x=\frac 1 {\sqrt{2}} \left(\begin{array}{ccc} 0 & 1 & 0 \\ 1 & 0 & 1 \\ 0 & 1 & 0 \end{array}\right),\quad
   \hat S_y=\frac i {\sqrt{2}} \left(\begin{array}{ccc} 0 & -1 & 0 \\ 1 & 0 & -1 \\ 0 & 1 & 0 \end{array}\right),\quad
   \hat S_z=\left(\begin{array}{ccc} 1 & 0 & 0 \\ 0 & 0 & 0 \\ 0 & 0 & -1 \end{array}\right).
   \label{eqSpin1}
\end{equation}

From this, we obtain a unique physical encoding $\varphi^{(S')}$ of quantum states. For example, the pure state $\rho_{\rm phys}=|\psi\rangle\langle\psi|_{\rm phys}$ that gives
unit probability of ``spin-up'' in $z$-direction will have
\[
   \varphi^{(S')}(\rho_{\rm phys})={\rm diag}(1,0,0).
\]
This algorithm of constructing $\varphi^{(S')}$ from $\varphi^{(S)}$ had one arbitrary choice, namely how to describe $\hat S_{\vec n}(S')$ in terms of concrete $3\times 3$
matrices. This is an arbitrary choice among all encodings (they are all related by unitary conjugation and possibly a transpose; other attempts of encoding
will be in conflict with the observed measurement statistics),
and part of the classical information that is communicated from Alice to Bob in the protocol. Fixing this data, the resulting encoding map $\varphi^{(S')}$ depends on $\varphi^{(S)}$: choosing another $\varphi^{(S)}$ in the first place
will select another physical device as the appropriate $\hat S_z(S)$-measurement, for example, which leads to another physical observable $\hat S_z(S')$ and
to another pure physical quantum state $|\psi\rangle_{\rm phys}$ on $S'$ that will be considered to be the $(+1)$-eigenstate of spin in $z$-direction.

Therefore, changing the encoding of $S$ by a rotation will have an effect on the encoding of $S'$; in other words, we get a projective spin-$1$ representation of ${\rm SO}(3)$ on
$S'$ (as claimed in Theorem~\ref{TheQubitNormed}), and this representation is not trivial.
\end{example}

The spin example uses the background knowledge that there is a notion of underlying $3$-space, constituting a mechanism
to set up these universal Stern-Gerlach devices in the first place. However, there is a way to argue that these universal Stern-Gerlach
devices should exist, without directly resorting to spatial degrees of freedom. The idea is that we can use a large number $N$ of
electron spin qubits to build a magnet. The ensemble of spins should be in a coherent state $|\vec n\rangle^{\otimes N}$, and then
another particle $S'$ can interact with that system via a simple interaction Hamiltonian~\cite{DakicBrukner3D}. Thus, the electron spin qubit's
quantum state can define a frame of reference which is transferred to other quantum systems via interaction. This suggests, but does not
necessitate the interpretation of $\vec n$ as a spatial direction.

This idea is very similar to von Weizs\"acker's suggestion that the symmetries of the elementary binary quantum alternative should
be identical to the symmetries of space~\cite{Weizsaecker}, and it resembles the distinguished role of the qubit for the structure
of quantum theory~\cite{Hardy01,Masanes11,Chiribella2011,Hardy11,Axioms2013,Hoehn2014,hw2}. It has recently been used to argue in a different context why
space must have three dimensions if a certain operational interplay between quantum theory and space is to hold~\cite{MM3D,DakicBrukner3D}.

We now turn to photon polarization $S''$. Given a photon with momentum $\vec p$ {(as seen in some reference frame)}, the photon spin must be oriented either parallel
or antiparallel to $\vec p$. Set $\vec p^0:=\vec p/|\vec p|$, and consider the spin observable $\hat S_{\vec p^0}$. This observable is
well-defined on the photon $S''$, and it can be written $\hat S_{\vec p^0}(S'')=|R\rangle\langle R|-|L\rangle\langle L|$, where $|R\rangle$
and $|L\rangle$ denote the left- and right-circular polarization states of the photon. Clearly, the same observable can be defined on
the spin-$1$ particle $S'$, which will be a $3\times 3$ matrix $\hat S_{\vec p^0}(S')$.

We claim that $\hat S_{\vec p^0}$ is universally measurable on $S'$ and $S''$, as already indicated by the notation. Namely, a concrete way to measure
both systems in the same device is via photon absorption by an atom in its electronic ground state, as described in~\cite{Suter,CCD}.
Conservation of angular momentum will force the atom's electron from the ground state to an $\ell=1$ excited state, with magnetic
quantum number $m_\ell=\pm 1$ corresponding to the photon's spin quantum state, $|L\rangle$ or $|R\rangle$. That is, the result of the absorption will effectively be
the transfer of the photonic quantum information on $S''$  to a quantum system $S'$. Finally, the observable $\hat S_{\vec p^0}(S')$ can be measured in a
Stern-Gerlach-like device, at least in principle.

Indeed, transmission of arbitrary quantum states from photon polarization qubits to energy levels of atoms are currently performed in many concrete
experiments, e.g., see~\cite{Kurz,kalb}. For example, this involves the experimental implementation of a photon-atom quantum gate \cite{reiserer,duan} whose construction is unambiguous once the spatial directions are fixed. Accordingly, two agents who have already agreed on the description of spatial directions via spin-$1$ systems $S'$ would have an unambiguous way of agreeing via classical communication on the construction and implementation of this photon-atom quantum gate. This would enable the two agents to also agree on the description of photon qubit states. It relies on the fact that, on the one hand, the photonic $|L\rangle$ and $|R\rangle$ states and their relative phase, and, on the other hand, also the spin-$1$ $|+1\rangle$ and $|-1\rangle$ states and their relative phase all have a geometric meaning.

More precisely, since superpositions of $|L\rangle$ and $|R\rangle$ are preserved {in the transmission from photons to atoms},
we can use the tomographic completeness of Stern-Gerlach measurements of spin in all directions~\cite{Amiet} to
effectively measure any observable with support on ${\rm span}\{|+1\rangle,|-1\rangle\}$ on quantum system $S'$, not only $\hat S_{\vec p^0}(S'')$.
This yields a set of observables that is universally measurable on $S'$ and $S''$ and tomographically complete on $S''$. Therefore, the universal measurability graph has an arrow from $S'$ to $S''$.

We have assumed in Subsection~\ref{SubsecFullLab} that the observables $\hat M(S')$ uniquely define the observables $\hat M(S'')$;
in other words, there is a \emph{unique} interpretation of these observables in terms of a physical quantity $\hat M$. This shows that
the other observables on $S'$ -- those that are not fully supported on ${\rm span}\{|+1\rangle,|-1\rangle\}$ -- cannot have counterparts
on $S''$. Therefore, we have already found the maximal set of universally measurable observables on $S'$ and $S''$, and no such set can be tomographically
complete on $S'$. Thus, there is no arrow in the universal measurability graph from $S''$ to $S'$.

Theorem~\ref{TheQubitNormed} also tells us that there is no representation of the rotation group on the photon polarization qubit,
in contrast to the spin-$1/2$ system $S$ and the spin-$1$ system $S'$. Instead, one expects to find a representation of the
subgroup of ${\rm SO}(3)$ that preserves the set of $(S',S'')$-co-measurable observables. These are exactly the rotations that stabilize
the photon momentum vector $\vec p$ -- in other words, the subgroup equivalent to ${\rm SO}(2)$ which rotates the transversal photon polarization vector.

\subsection{Two causes of disagreement: active and passive}
So far, we have motivated the ``information gap'' between Alice and Bob by imagining that they reside in different, distant laboratories, and have
never met before. Under this assumption, we have argued in Subsection~\ref{SubsecFullLab} that their descriptions of local quantum physics in their laboratories must be related
by an element of $\g_{\min}={\rm O}(3)$.

Concretely, we can think of Alice and Bob as talking on the phone, Alice sending a request to Bob in terms of classical information, and Bob responding
in terms of a physical system that he sends back, as depicted in Figure~\ref{fig_scenario1}. Clearly, the classical information that Alice sends to Bob
must be encoded into \emph{some} physical system as well, that serves as the signal carrier. However, we can imagine that they are using a physical (quantum) system $S'$ that carries at least
a partial natural choice of Hilbert space basis. For example, Alice can send information bitwise, encoding a zero into the ground state, and a one into
an excited state of a two-level system; or encoding it into the \emph{relation} between the orthogonal basis elements of two successive quantum systems.

In addition to this ``passive'' picture, we can also imagine an ``active'' scenario: suppose that Alice and Bob \emph{have} actually met before, and have used
this encounter to calibrate their measurement devices and synchronize their frames of reference. But imagine that their laboratories have become separated
afterwards, and subjected to all kinds of physical influences. Concretely, think of Alice and Bob and their laboratories as traveling in spaceships to distant parts of
the galaxy, under the influence of all kinds of gravitational fields.

In this case, their descriptions of local physics of each others' laboratories will have become desynchronized, and they may recognize it once they try to implement the information-theoretic
scenario of Figure~\ref{fig_scenario1}. However, in this case, there is a difference to the earlier ``passive'' scenario: namely, one would expect that
their ``information gap'' is characterized by an element of the \emph{implementable subgroup} of $\g_{\min}$ -- that is, in this case, of ${\rm SO}(3)$.
This is because the universe \emph{has} actually implemented the corresponding transformation, by acting on their laboratories.

Indeed, this is consistent with our result: different observers may become twisted relative to each other (described by a rotation $R\in{\rm SO}(3)$), but usually not reflected relative to each other.

\section{Physical quantum states and the Lorentz group}
\label{sec_phys}

\subsection{Types of quantum systems and universal measurement devices}
\label{SubsecTypes}

In the previous section, we have taken an abstract operational point of view on quantum states inspired by quantum information theory.
In this picture, a state of a quantum $N$-level system is merely a concise catalogue of probabilities of the outcomes of all possible measurements
that can be performed on the quantum system. For example, if a state $\rho=\sum_{i=1}^N \lambda_i |i\rangle\langle i|$ (assuming
non-degenerate spectrum) is measured in its
eigenbasis, then the probability to obtain outcome $i$ is $\lambda_i$, but the outcome itself is not considered to have any particular
physical meaning. Any observable of the form
\[
   \hat M=\sum_{i=1}^N m_i |i\rangle\langle i|
\]
for any choice of $m_1,\ldots,m_N$ can be measured by the same physical device with outcomes that merely differ by their classical labels $m_i$.
This point of view is particularly common in quantum information theory, and it seems especially appropriate
in the case of destructive measurements where the physical system is annihilated on detection.

However, in many situations, \emph{measurement outcomes carry concrete physical meaning}, in the sense that the outcome describes
a specific physical property of the physical system after the measurement. In this case, the eigenvalue is not just a classical label, but describes a physical post-measurement property (say, a particle's kinetic energy in some units like Joule). We have already seen in Subsection~\ref{SubsecFullLab}
that actual physical properties of an observable (for example the property of being measurable by universal devices) can have important structural consequences.
Therefore, we should analyze how the conclusion of the previous section are modified if we take into account that measurement outcomes have actually a ``size'' which can be compared to other outcomes.

Suppose we have a physical quantity which can in principle take one of infinitely, maybe continuously many values (such as energy). Then even if we have
an effectively finite-dimensional quantum system (such as a superposition of only two energy levels in an atom), this quantum system will be a subspace
of a much larger, typically infinite-dimensional Hilbert space or operator algebra which describes the laboratory as a whole {(once the necessary localization, e.g., within quantum field theory is possible)}. We will now argue that no matter
which fundamental theory (say, what specific quantum field theory) we assume to hold, this simple fact will have the universal consequence that finite-dimensional
quantum subsystems will come in different \emph{types}, even if they have the same Hilbert space dimensionality. A simple example illustrates this fact.

Imagine a Stern-Gerlach apparatus which performs a spin measurement on spin-$1/2$ particles. For this, the particles of mass $m$ and velocity $v$ will enter an inhomogeneous
magnetic field which defines a quantization axis $z$, and then spread into two beams corresponding to the two possible values of the spin in $z$-direction,
finally hitting a screen (let us assume for now that $v_z=0$ before the particles reach the magnetic field).
The particles will hit the screen in a small vicinity along the $z$-direction either up- or downwards from the initial beam axis, which constitutes the two possible outcomes ``spin up'' or ``spin down''.
We can associate an observable, $\Delta z$, to this
experiment, with eigenvalues corresponding to the two possible spatial displacements (in units of length) relative to the original beam axis and along the quantization axis.
The two eigenvalues (differing only in sign), corresponding to the two possible deflections\footnote{We will come back to the Stern-Gerlach device
later. It will turn out that our abstract derivation which follows describes the deflection more precisely in terms of the acceleration that the particle undergoes in the inhomogeneous field rather than in units of lengths. Nevertheless, qualitatively, the illustration in terms of $\Delta z$ is simpler and conveys the same motivation for the general abstract argumentation. We thus stick to it here.
}, are proportional to $1/(mv^2)$, as well as the magnetic moment and the magnetic field gradient. We emphasize that, in contrast to before, we now have classical scale in the game, described in this case by units of length.
\begin{figure*}[!hbt]
\begin{center}
\setlength{\fboxrule}{1pt}
\fbox{\includegraphics[angle=0, width=10cm]{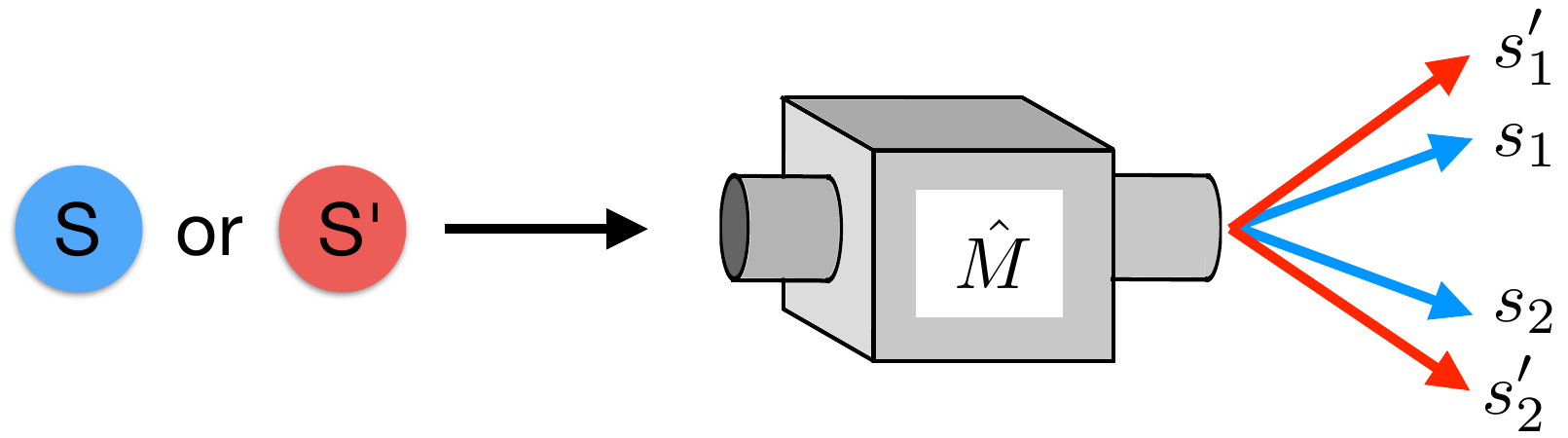}}
\caption{If measurement outcomes (i.e.\ eigenvalues of observables) are physical quantities, then there will typically be different ``types'' of systems $S$ and $S'$. They may have the same Hilbert space dimensionality, but still lead to measurement of different observables if input into the same measurement device. This leads us to refine the notion of ``universal measurement device'' from Section~\ref{SecO3}.}
\label{fig_universal2}
\end{center}
\end{figure*}

Thus, two electrons $S$ and $S'$ with different velocities will make the same fixed Stern-Gerlach device perform
different kinds of measurements (cf.\ Figure~\ref{fig_universal2}). More precisely, even though the same physical degree of freedom is probed in both cases (namely the quantum bit carried by
the electron's spin), the actual qubit observable $\Delta z$ differs between both cases {in the sense that the two sets of possible outcomes do not coincide. In fact, $\Delta z$ can be viewed as a family of {\it hybrid} observables with a quantum and a classical part: the quantum part, the spin, determines whether the particles are deflected `up' or `down', while the classical part (parametrizing the family) determines {\it how strongly} the systems are deflected}. Somehow, the qubits come in different \emph{types}
(in this case corresponding to different velocities $v$), and the measurement device responds differently to the different types of qubits.

Another possible source of having different types of quantum states is that \emph{there are different sorts of physical systems} that carry them, such as different sorts of particles.
For example, we can use other types of neutral atoms with an unpaired electron, or in principle even think of replacing the electron by a muon.

Once again, we are led to consider measurement devices which are able to measure different kinds of quantum systems, similarly as the \emph{universal measurement devices} of Section~\ref{SecO3}. But now we have additional complexity: if the outcomes (that is, eigenvalues) have physical meaning, then different observers may not agree on their description. For example, if (as in the example above) the outcomes have a unit of length, then Alice and Bob may use different units to describe length, and in this sense disagree on the actual numerical value of the eigenvalues.

\begin{figure*}[!hbt]
\begin{center}
\includegraphics[angle=0, width=9cm]{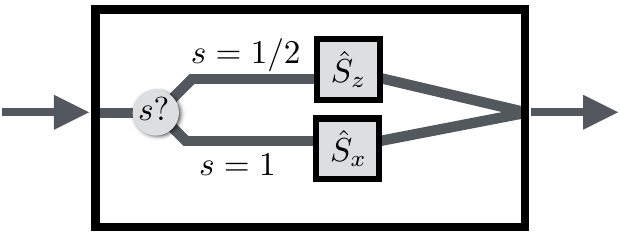}
\caption{Not a universal measurement device. By construction, the observables that this device measures on systems of spin-$1/2$ and spin-$1$ are physically different. Since there is in general no way to ``look inside the black box'', a definition of a universal measurement device has to appeal to consistency conditions among \emph{several} universal measurement devices. These conditions will rule out this device from the family of universal spin measurement devices that {we have} in our universe (see main text).}
\label{fig_notuniversal}
\end{center}
\end{figure*}

In order to analyze this more complicated situation, we need to say in more detail what it means for a measurement device to perform ``the same measurement'' on different quantum systems. In Assumptions~\ref{AssUniversal3} of Section~\ref{SecO3}, we have only assumed a certain continuity property; now we have to be more specific and extend our definition. For example, we would intuitively believe that a device as in Figure~\ref{fig_universal2} performs the same measurement on particles of different velocities. On the other hand, we would say that a device like the one in Figure~\ref{fig_notuniversal} is \emph{not} a good universal measurement device in that sense: by construction, it performs a different measurement for spin-$1/2$ particles than it does for spin-$1$ particles. But how can we know this if we do not know the inner workings of a measurement device, but only their effects?

The main idea (which is also implicit in Assumptions~\ref{AssUniversal3}, as we shall soon see) is not to consider single measurement devices separately, but to think of a \emph{family of universal measurement devices} which are related to each other by consistency conditions. For example, suppose that we have identified Stern-Gerlach devices that measure the spin in $x$- and $y$-direction, and do so for spin-$1/2$ particles $S$ as well as for spin-$1$ particles $S'$. If these {are} universal measurement devices from one family, then the device in Figure~\ref{fig_notuniversal} cannot be a universal measurement device from the same family too: the former would assign an observable $\hat S_x(S)$ to $\hat S_x(S')$, while the device in Figure~\ref{fig_notuniversal} would assign an observable $\hat S_z(S)$ to $\hat S_x(S')$, contradicting uniqueness of Assumptions~\ref{AssUniversal3}. {For example,} an observer who has a notion of directions and thereby of spin observables too can rule out universality of the device of Figure~\ref{fig_notuniversal} by comparing it to Stern-Gerlach measurements.

But which consistency conditions should we choose to determine whether two devices are from the {\it same} family of universal measurement devices? If the outcomes of a measurement are physical scalars, then we should consider the most important structural and operational property of the real numbers: that they form an ordered field. In particular, we have a notion of \emph{comparing} the values of a physical quantity on two different physical systems. It seems that the notion of ``one quantity being larger than another'' is one of the most important prerequisites to define physical quantities (like lengths or masses) operationally.

This leads us to our first new postulate: \emph{comparison of outcomes should be system-independent}. More concretely, suppose that $\hat M_1(S)$ and $\hat M_2(S)$ are devices that measure two observables on quantum system $S$, such that the expectation value of $\hat M_2(S)$ is always at least as large as that of $\hat M_1(S)$, no matter what state we send in. Then we expect that this property remains true if we measure quantum system $S'$ in the same device; that is, we obtain inequality~(\ref{eqOrder}) below.

For convenience, we add one more new postulate (and we will briefly discuss dropping it later). It encodes the idea that the outcomes of the measurement are taken \emph{relative to their values before the measurement}. In particular, an outcome corresponding to eigenvalue $\lambda=0$ corresponds to an outcome where \emph{the scalar physical quantity of interest after the measurement is identical to its value before the measurement}. In particular, the observable $\hat M=0$ (the zero matrix) is interpreted as \emph{performing no measurement at all}. We assume that this property is stable across the different quantum systems that are measured in the same device.

\begin{definition}[Universal measurement devices, extension of Assumptions~\ref{AssUniversal3}]
\label{DefUniversal4}
$\strut$\newline
We assume that there exists a family of ``universal measurement devices'' that can measure ``the same observable'' $\hat M$ on several different quantum systems $S,S',\ldots$, characterized by corresponding observables $\hat M(S),\hat M(S'),\ldots$. If quantum systems $S$ and $S'$ can be measured in a single universal measurement device, then we assume that $\hat M(S)$ uniquely determines $\hat M(S')$, and vice versa. Moreover, we assume that this interdependence is continuous in both directions.

Furthermore, we make the following assumptions as motivated above. If $\hat M_1$ and $\hat M_2$ are two universal measurement devices in which we can measure both $S$ and $S'$, then
\begin{equation}
   \hat M_1(S)\leq \hat M_2(S)\quad\Longrightarrow \quad\hat M_1(S')\leq \hat M_2(S') \qquad \mbox{(``order preservation'')}.
   \label{eqOrder}
\end{equation}
Furthermore, there is a unique ``zero measurement'' as explained above. That is,
\begin{equation}
   \hat M(S)=0 \quad \Longrightarrow\quad \hat M(S')=0.
   \label{eqZero}
\end{equation}
\end{definition}
Before we discuss these conditions in more detail, we take courage by observing that Stern-Gerlach spin measurements satisfy these conditions. First, think of observables $\hat M_i$ that measure the absolute value of the deflection, $|\Delta z|$, for a smaller ($i=1$) resp.\ larger ($i=2$) magnitude of the magnetic field gradient. Regardless of the precise incident system, the expectation value of $\hat M_1$ will always be less than the expectation value of $\hat M_2$, satisfying~(\ref{eqOrder}). The ``zero measurement'' in~(\ref{eqZero}) corresponds to the magnetic field being switched off which amounts to zero deflection for all incident systems and therefore no spin measurement. Later on, in Subsection~\ref{sec_SG}, we will show more thoroughly that a proper relativistic description of Stern-Gerlach measurements on spin-$1/2$ particles, taking their state of motion into account, will follow these conditions as well.

We can also show directly that the spin measurements of Section~\ref{SecO3} (with fixed eigenvalues independent of the deflection $\Delta z$) satisfy Definition~\ref{DefUniversal4} -- as long as they encompass a universal identity observable $\mathbf{\hat 1}$. Recall Theorem~\ref{TheQubitNormed}, and consider the case that there is a ``root qubit'' ($d=2$) such that the symmetry group becomes $\g_{\rm min}={\rm O}(3)$. In this case, the different quantum systems $S,S',\ldots$ correspond to the internal degrees of freedom of particles with different spin, while their external attributes -- such as the state of motion -- are irrelevant. The identity $\mathbf{\hat1}$ completes the spin operators to a full basis for all qubit observables.

\begin{lemma}[Consistency with Section~\ref{SecO3}]
\label{LemConsistency}
Denote by $S_s$ the quantum system that corresponds to the internal degree of freedom of a spin-$s$ particle, $s\in\{\frac 1 2,1,\frac 3 2 ,2,\frac 5 2,\ldots\}$. Use the abbreviation $\hat M(S_s)=:\hat M^{(s)}$. Define the observables $\hat S_x^{(s)}$, $\hat S_y^{(s)}$, $\hat S_z^{(s)}$ as arbitrary spin-$s$ representations of the angular momentum commutation relations $[\hat S_x^{(s)},\hat S_y^{(s)}]=i \hat S_z^{(s)}$ etc. Furthermore, define $\mathbf{\hat 1}^{(s)}:=s\mathbf{1}$. In particular, in some suitable choice of basis, this gives the Pauli matrices
\[
   \mathbf{\hat 1}^{(1/2)}=\frac 1 2 \left(\begin{array}{cc} 1 & 0 \\ 0 & 1 \end{array}\right),\quad
   \hat S_x^{(1/2)}=\frac 1 2 \left(\begin{array}{cc} 0 & 1 \\ 1 & 0 \end{array}\right),\quad
   \hat S_y^{(1/2)}=\frac 1 2 \left(\begin{array}{cc} 0 & -i \\ i & 0 \end{array}\right),\quad
   \hat S_z^{(1/2)}=\frac 1 2 \left(\begin{array}{cc} 1 & 0 \\ 0 & -1 \end{array}\right).
\]
Every observable $\hat M^{(1/2)}$ on the two-dimensional spin-$1/2$ system can be written as a linear combination with real coefficients
\[
   \hat M^{(1/2)}=\alpha \mathbf{\hat 1}^{(1/2)}+\alpha_x \hat S_x^{(1/2)}+\alpha_y \hat S_y^{(1/2)}+\alpha_z \hat S_z^{(1/2)}.
\]
Define $\hat M^{(s)}$ for all $s\neq \frac 1 2$ by linear extension of the above, i.e.
\begin{equation}
   \hat M^{(s)}=\alpha \mathbf{\hat 1}^{(s)}+\alpha_x \hat S_x^{(s)}+\alpha_y \hat S_y^{(s)}+\alpha_z \hat S_z^{(s)}.
   \label{eqLinComb}
\end{equation}
A family of spin measurement devices satisfying this prescription also fulfills Definition~\ref{DefUniversal4} above and is thus universal. In particular, we have order preservation as in~(\ref{eqOrder}).
\end{lemma}
This prescription describes the spin measurements in our universe. The (idealized) Stern-Gerlach devices of Section~\ref{SecO3}, each measuring solely the spin in a fixed direction on different sorts of particles (but not their precise deflection), are thus universal.
This requires a universal observable $\mathbf{\hat 1}$ as above which renders (\ref{eqOrder}) non-trivial (in fact, the choice above turns out to be the unique possibility in order to have $s$-independent positivity). Note that $\mathbf{\hat 1}(s)$ is a function of the total spin (which would be $\sqrt{s(s+1)}\mathbf{1}$). 
\begin{proof}
It is clear that linearity implies~(\ref{eqZero}); it also implies that we only have to check $\hat M^{(s)}\geq 0 \Rightarrow \hat M^{(s')}\geq 0$ to prove~(\ref{eqOrder}), i.e.\ that positivity is preserved among the different spins. Consider~(\ref{eqLinComb}) in the special case that $\alpha=0$ and $\alpha_x^2+\alpha_y^2+\alpha_z^2=1$. Then $\hat M^{(s)}=\hat S_{\vec n}$ is another spin-$s$ operator, namely the one that describes a spin measurement in direction $\vec n=(\alpha_x,\alpha_y,\alpha_z)$. Thus, its smallest eigenvalue is $(-s)$. Then it is clear that the smallest eigenvalue of $\hat M^{(s)}$ in general is $\alpha s -\sqrt{\alpha_x^2+\alpha_y^2+\alpha_z^2}s$. Therefore $\hat M^{(s)}\geq 0$ if and only if $\alpha\geq \sqrt{\alpha_x^2+\alpha_y^2+\alpha_z^2}$. This condition is independent of $s$.
\end{proof}
This raises the question whether we could have an analogous consistency result also in the case that $d\geq 3$ in Theorem~\ref{TheQubitNormed}, in which case we would have the symmetry group ${\rm PUA}(d)$ instead of ${\rm O}(3)$. If the answer was negative, then this would yield additional justification of the assumption that there is a ``qubit root'' with $d=2$ (or, in other words, this would single out the qubit as ``special''). We do not currently know the answer to this question.

Some quantum systems will share all their observables with other quantum systems:
\begin{definition}[Equivalent quantum system]
\label{DefEquivalence}
We say that two quantum systems $S$ and $S'$ are \emph{equivalent} if and only if all observables of $S$ are universally measurable on both $S$ and $S'$, and vice versa. We denote the equivalence classes in bold face font $\mathbf{S}$.
\end{definition}
Due to lack of a better wording, we will also call the equivalence classes $\mathbf{S}$ ``quantum systems'' (in fact, they are collections of many quantum systems). Surprisingly, Definition~\ref{DefUniversal4} turns out to be quite restrictive:
\begin{lemma}
\label{LemTypeTransform}
Let $S$ and $S'$ be equivalent quantum systems, then $\dim S=\dim S'$. Ignoring condition~(\ref{eqZero}) in Definition~\ref{DefUniversal4} implies that there is an invertible complex $N\times N$ matrix $X$ such that for all observables $\hat M$,
\begin{equation}
   \mbox{either } \qquad\hat M(S')=X^\dagger\hat M(S)X+Y\qquad\mbox{ or } \qquad\hat M(S') = X^\dagger\hat M(S)^T X+Y,
   \label{eqPoincare}
\end{equation}
where $Y$ is some Hermitian $N\times N$ matrix. If condition~(\ref{eqZero}) is taken into account (and therefore Definition~\ref{DefUniversal4} is considered in its complete form), then $Y=0$.
\end{lemma}
\begin{proof}
If $S$ and $S'$ are equivalent, then Definition~\ref{DefUniversal4} implies that the spaces of Hermitian matrices over $S$ resp.\ $S'$ are homeomorphic, therefore they have the same dimensions, and hence $\dim S=\dim S'$.
It is easy to check that maps of the form given above satisfy all conditions of Definition~\ref{DefUniversal4}. Conversely, note that all continuous bijective order-preserving maps on the Hermitian $N\times N$ matrices which preserve the zero matrix are of the form $\hat M\mapsto X^\dagger\hat M X$ or $\hat M\mapsto X^\dagger \hat M^T X$ with $X$ some invertible complex $N\times N$ matrix, see~\cite{Semrl,Molnar}.
\end{proof}
This lemma shows that some further natural consistency conditions among the different quantum systems hold true automatically for equivalent systems, even though we have not postulated them directly in Definition~\ref{DefUniversal4}. For example, we have linearity:
\begin{equation}
   \hat M(S)=\lambda_1\, \hat M_1(S)+\lambda_2\hat M_2(S) \quad\Longrightarrow \quad
   \hat M(S')=\lambda_1\, \hat M_1(S')+\lambda_2\hat M_2(S').
   \label{eqLinearity}
\end{equation}
In Subsection~\ref{SubsecPSL}, we will discuss the communication scenario of Alice and Bob as in Figure~\ref{fig_scenario1} in the context of equivalent quantum systems. As we will discuss there, Alice and Bob will agree on equivalence classes $\mathbf{S}$ of quantum systems, but not on the question which one of the equivalent quantum systems ($S\in\mathbf{S}$ or $S'\in\mathbf{S}$ etc.) is to be sent. Nevertheless, (\ref{eqLinearity}) can be interpreted as saying that \emph{Alice and Bob agree on the linearity structure of observables}. And this is quite intuitive: think of two observers who measure distances in meters and miles, respectively. Consider the distances $d_1=4$m, $d_2=6$m, $d_3=10$m and $d_4=24$m. Then both observers will agree on the fact that $d_1+d_2=d_3$ -- this equation remains true after conversion from meters into miles. However, the fact that the numerical value of $d_4$ is the \emph{product} of the numerical values of $d_1$ and $d_2$ is something that the observers will not in general agree upon. In the next subsection,
we will also see that this property is even more important when we turn from the Heisenberg to the Schr\"odinger picture: the ``dual'' property of~(\ref{eqLinearity}) establishes the observer-independence of the probabilistic structure.

At the beginning of this section (see also Figure~\ref{fig_universal2}), we have motivated our approach by saying that quantum systems may come in different ``types'', making one and the same measurement device perform different kinds of measurements with potentially different outcome eigenvalues. We are now in a position to write down a formal definition of types. We will say that $S$ and $S'$ are of the same type if \emph{for every state $\rho_S$ of system $S$, there is a state $\rho_{S'}$ of system $S'$
that yields exactly the same statistics on all measurements $\hat M$ (and vice versa)}. We can interpret this as saying that $S$ and $S'$ are identical quantum systems, but carry different choices of Hilbert space bases. This leads us to the following definition.
\begin{definition}\label{def_type} \emph{\bf(Type)}
Two quantum systems $S$ and $S'$ are said to be of the same \emph{type} if they are equivalent in the sense of Definition~\ref{DefEquivalence}, and if either $\hat M(S')=U^\dagger\, \hat M(S)\, U$ or $\hat M(S')=U^\dagger\, \hat M(S)^T\, U$ with some fixed unitary $U$. Otherwise, we say that $S$ and $S'$ are of distinct \emph{types}.
\end{definition}
We can also nicely characterize this in a different way:
\begin{lemma}
\label{LemTypeEigenvalues}
Two equivalent quantum systems $S$ and $S'$ are of the same type if and only if the outcome eigenvalues of all measurements are identical.	
\end{lemma}
\begin{proof}
Clearly, if $S$ and $S'$ are of the same type, then $\hat M(S)$ and $\hat M(S')$ are related by a unitary (and possibly a transpose) which preserves all eigenvalues. Conversely, the results in~\cite{Jafarian} show that the map $X^\dagger \bullet X$ which relates $\hat M(S)$ and $\hat M(S')$ according to Lemma~\ref{LemTypeTransform} is spectrum-preserving if and only if $X^\dagger=X^{-1}$, i.e.\ if $X$ is unitary and thus $S$ and $S'$ are of the same type.
\end{proof}
Therefore, in order to specify $\hat M(S)$ for a given quantum system $S$, we need to know two things: first, the type of the system $S$, and second, the choice of basis in the corresponding Hilbert space, or more generally the way that quantum states are described as density matrices.

The observables of different types of systems are thus related by conjugation with {\it non-unitary} matrices. The interpretation of this non-unitarity is clear: different types of systems produce distinct outcomes with distinct physical values in the {\it same} measurement devices (cf.\ Figure~\ref{fig_universal2}). Accordingly, their corresponding observables must feature distinct eigenvalues such that a transformation among them cannot be unitary.

\subsection{From Heisenberg to Schr\"odinger: unnormalized quantum states}
\label{SubsecSchroedinger}
Before we discuss the communication task between Alice and Bob under these generalized circumstances, we need to clarify what the type transformations imply for quantum states. In standard quantum theory, states are dual to observables and symmetry transformations can be performed on either.  So far, we have taken the different types of systems into account by saying that a measurement $\hat M$ is described by an observable $\hat M(S)$ which depends on the quantum system $S$. Translating between equivalent systems implies, according to Lemma~\ref{LemTypeTransform}, a transformation of the observables. However, we can also use a ``dual'' formalism, which is similar to the change from the Heisenberg to the Schr\"odinger picture, and apply the transformation equivalently to states. But how can we reconcile the non-unitarity of type transformations with a meaningful notion of state normalization? As we shall now argue, one can translate the type freedom into a parametric freedom in the state normalization which ultimately implies that degrees of freedom of distinct type live on distinct Hilbert spaces. Not only will this be useful to discuss the communication scenario of Figure~\ref{fig_scenario1} in Subsection~\ref{SubsecPSL}, but it will also allow us to show a direct connection to relativistic quantum theory in Subsection~\ref{sec_SG}.

For any given quantum system (equivalence class) $\mathbf{S}$, fix some reference system $S_0\in\mathbf{S}$. Suppose that $S\in\mathbf{S}$ is any other quantum system such that Lemma~\ref{LemTypeTransform} yields the transformation equation without transpose, i.e.\ $\hat M(S)=X^\dagger\hat M(S_0)X$. Then, for any quantum state $\rho$ on $S$,
\[
   \tr(\rho \hat M(S))=\tr(\rho X^\dagger \hat M(S_0)X)=\tr(X \rho X^\dagger \hat M(S_0)).
\]
This allows us to use the following alternative formalism: every measurement $\hat M$ is described by a \emph{fixed} observable $\hat M(S_0)$.
However, quantum states $\rho$ of other equivalent systems $S\in\mathbf{S}$ are instead described by pseudo quantum states
$X \rho X^\dagger$, where $X$ is the complex matrix
such that $\hat M(S)=X^\dagger\hat M(S_0)X$ in the ``Heisenberg'' formalism we had before. In the case where $\hat M(S)=X^\dagger\hat M(S_0)^T X$, the pseudo quantum state
is $\bar X \rho^T \bar X^\dagger$, where $\bar X$ denotes the complex conjugate of $X$.

Note that this ``dualization'' would be impossible if we had $Y\neq 0$ in Lemma~\ref{LemTypeTransform}. Thus, condition~(\ref{eqZero}) is necessary
for translating from type-dependent observables to type-dependent quantum states if we want to keep the usual Born rule (trace rule) formalism.

Pseudo quantum states are positive semidefinite, but not necessarily normalized; their trace can be any positive real number.
At first sight, this seems to be in conflict with the usual postulates of quantum mechanics, but in fact it is not. Every $N$-level pseudo quantum state $\rho$ is an operator on $\C^N$. However, this vector space becomes a Hilbert space only with a choice of inner
product. It turns out that choosing an inner product different from the canonical one turns pseudo quantum states into quantum states,
and thus admits the usual probabilistic interpretation of quantum states. Moreover, this choice of inner product is necessary and sufficient to have a valid ``Schr\"odinger picture'' that makes identical physical predictions as the ``Heisenberg picture'' that we have used in the previous subsections.

\begin{observation}
\label{Obs1}
By choosing a fixed reference quantum system $S_0$, we can denote every equivalent system $S$ by $S_X$ or $S_X^T$; we use the
former notation if $\hat M(S)=X^\dagger \hat M(S_0)X$, and the latter if $\hat M(S)=X^\dagger \hat M(S_0)^T X$.

For any given equivalent quantum system $S\in\mathbf{S}$, consider the polar decomposition
\[
   X=\left\{
      \begin{array}{cl}
         RU & \mbox{if }S=S_X\\
         R^T U & \mbox{if }S=S_X^T
      \end{array}
   \right\},
   \qquad\mbox{with }R>0\mbox{ and }U^\dagger U=\mathbf{1}.
\]
This decomposition is unique~\cite{HornJohnson}. Note that\footnote{To see this,
suppose $X^\dagger \hat M X = (X')^\dagger \hat M X'$ for all $\hat M$.
The special case $\hat M=(X X^\dagger)^{-1}$ implies that $XX^\dagger=X'(X')^\dagger$, and thus $R=R'$. Using this and substituting $\hat M \mapsto R^{-1}\hat M R^{-1}$
gives $U^\dagger \hat M U = (U')^\dagger \hat M U'$ for all $\hat M$, which implies $U'=e^{i\theta}U$.
} $S_X=S_{X'}$
(resp.\ $S_X^T=S_{X'}^T$) if and only if $X'=e^{i\theta}X$ for some $\theta\in\R$ (or equivalently if $U'=e^{i\theta}U$ and $R=R'$).
Pseudo quantum states on $S_X$ or $S_X^T$ are quantum states
on the Hilbert space $\mathcal{H}_R$, which is the complex vector space $\C^N$ with inner product
\[
   \langle \varphi|\psi\rangle_R:=\langle \varphi|R^{-2}|\psi\rangle,
\]
and the unitary $U$ signifies the different choice of basis in this Hilbert space, compared to the choice of basis for $S_0$. Two equivalent quantum systems are of the same type if and only if their matrices $R$ agree; that is, if and only if they live on the same Hilbert space, $\mathcal{H}_R$.
If $S=S_X$ (no transposition), then
the matrices $X$, interpreted as maps $X:\mathcal{H}_{\mathbf{1}}\to\mathcal{H}_R$, are isometries of Hilbert spaces.
In the case where $S=S_X^T$, the map $|\psi\rangle\mapsto \bar X |\bar\psi\rangle$, where the bar denotes componentwise complex conjugation, is an antilinear isometry.
Normalization of pseudo quantum states $\rho$ is given by $\tr^R(\rho):=\tr(\rho R^{-2})=1$.
\end{observation}

Given pure quantum states $|\psi\rangle$ and $|\varphi\rangle$ on $S_X$, the pseudo quantum states are $|\psi'\rangle=X|\psi\rangle$
and $|\varphi'\rangle=X|\varphi\rangle$. With the given choice of inner product, we have
\[
   \langle \varphi'|\psi'\rangle_R=\langle \varphi|X^\dagger R^{-2} X|\psi\rangle = \langle \varphi| U^\dagger R R^{-2} R U |\psi\rangle = \langle \varphi|\psi\rangle.
\]
Thus, pseudo quantum states are indeed normalized with this choice of inner product, and they yield valid transition probabilities. On quantum system $S_X^T$, we instead have $|\psi'\rangle=\bar X|\bar\psi\rangle$, and we obtain analogously
\[
   \langle \varphi'|\psi'\rangle_R = \langle \psi|\varphi\rangle = \overline{\langle \varphi|\psi\rangle}.
\]
Since transition probabilities are given by the modulus square of an inner product, this still reproduces the probabilistic predictions of the standard quantum states in the Heisenberg picture of the previous subsection.
More generally, for systems $S_X$, we can check that the normalization of the pseudo quantum state
$\rho'=X\rho X^\dagger$ in the new inner product\footnote{We can describe the situation alternatively as follows. Given $\C^N$ as a vector space,
the choice of inner product $\langle \cdot,\cdot\rangle_R$ defines a set of Hermitian matrices, i.e.\ matrices $M$ with $\langle x,My\rangle_R=\langle Mx,y\rangle_R$.
These are the matrices for which $(R^{-2}M)^\dagger = R^{-2}M$, where the conjugate transpose is defined with respect to the standard inner product and its canonical orthonormal basis.
Similarly, the cone of positive semidefinite matrices becomes the set of matrices $M$ for which $R^{-2}M\geq 0$ with respect to the canonical inner product,
and the normalization condition becomes $\tr^R(M)\equiv \tr(R^{-2}M)=1$. All this is consistent with the Heisenberg picture: consider a quantum system $S$ in state $|\psi\rangle$. The expectation value of an observable $\hat M$ is $\langle \psi|\hat M(S)|\psi\rangle$, where $\hat M(S)=\hat M(S)^\dagger$.
If $S=S_X$, then this equals $\langle \varphi|R^2 \hat M(S_0)|\varphi\rangle_R$, where $|\varphi\rangle=X|\psi\rangle$, and $R^2\hat M(S_0)$ is Hermitian
in the Hilbert space $\mathcal{H}_R$.} equals the usual trace of $\rho$:
\ba
   \tr^R\rho'=\tr(X\rho X^\dagger R^{-2})=\tr(\rho X^\dagger R^{-2} X)=\tr\, \rho=1,
   \label{tra}
\ea
and the same result holds for systems $S_X^T$, using that $\tr(\rho^T)=\tr\,\rho$.

We will discuss the relation of these results to the physics of relativistic Stern-Gerlach devices in some detail in Subsection~\ref{sec_SG}. For the moment, let us simply anticipate a few points to see where this is going. As we have just seen, degrees of freedom of different type live on different Hilbert spaces, which means that we cannot construct quantum states describing their superposition and have to treat the type of a system \emph{classically}. In Subsection~\ref{sec_SG}, we will see that the different types of systems basically correspond to different momenta $p$ of particles in the Stern-Gerlach case. Thus, our situation is qualitatively different from quantum field theory, where superposition of different momentum modes is standard. However, this does not invalidate the current approach. Instead, our framework can be viewed as an approximation that is valid in certain physical situations.

For instance, in a spacetime context, this framework will apply to the special case of particles that can be approximately described to be in classical momentum states; the momentum-dependence of the Hilbert space inner product is identical to the one that was independently constructed in a relativistic description of the Stern-Gerlach experiment by Palmer et al.~\cite{Westman2}. We will come back to this later in Subsection~\ref{sec_SG}.

\subsection{The communication scenario with different types of systems}
\label{SubsecPSL}
We now return to the communication scenario of the previous sections. Again we have Alice and Bob as in Figure~\ref{fig_scenario1}.
Alice sends Bob a classical request, asking him to prepare a quantum system $\mathbf{S}$ in a quantum state $\rho$. Bob then tries to send the corresponding physical system back to Alice. We start with a simple situation that is the analogue of the situation considered in Subsection~\ref{SubsecTransmitFinite}: Alice and Bob aim at transmitting a single concrete $N$-level quantum system $\mathbf{S}$; no assumptions are made on any relation between different quantum systems $\mathbf{S}$ and $\mathbf{S'}$.

This formulation implicitly assumes the following:
\begin{assumptions}
\label{Ass4.8}
We assume that Alice and Bob can agree on the equivalence class $\mathbf{S}$ by classical communication (up to a subtlety discussed right before Theorem~\ref{TheQubitUnnormed}), but not in general on the physical system $S\in\mathbf{S}$ that is to be sent.\end{assumptions}
Note that the difficulty of picking out a quantum system $S$ from a continuum of systems (and describing it classically) was already implicit in Section~\ref{SecO3}, in the case of systems which carry only a subset of universally measurable observables (see the proof of Theorem~\ref{TheQubitNormed}).

For the time being, let us describe the situation in the Heisenberg picture.
The main communication problem -- if Alice and Bob have never
met before -- is that both may describe equivalent quantum systems in different ways. That is, Alice and Bob may have chosen two different ``reference systems''
$S_0^A\in\mathbf{S}$ and $S_0^B\in\mathbf{S}$ respectively, which means that they assign different invertible matrices $X_A$ and $X_B$ to the same quantum system $S$ (namely such that $\hat M(S)=X_A^\dagger \hat M(S_0^A)X_A$ and similarly for $B$, possibly with transpositions).
To distinguish between both cases with and without transpose, we introduce another variable $\mathcal{P}\in\{-1,+1\}$,
``parity'', which is set to $+1$ if there is no transpose, and $-1$ otherwise. Again, Alice and Bob may differ in the assignment of parity, in which case $\mathcal{P}_A\neq \mathcal{P}_B$.

Note that the communication scenario in the previous Section~\ref{SecO3} is a simpler version of this one: it applies to the case where Alice and Bob \emph{agree on the type of
system to be sent}, such that the only source of disagreement is the choice of Hilbert space basis, i.e.\
$\hat M(S_0^B)=U^\dagger \hat M(S_0^A) U$ for some unitary $U$ and all observables $\hat M$ (possibly with an additional transpose). With this additional knowledge, Alice's and Bob's descriptions differ by a unitary or antiunitary map,
and we recover the setting of Section~\ref{SecO3}. The additional complexity of \emph{this} section is that Alice and Bob do not in general agree on the description of eigenvalues of measurement outcomes; which is, according to Lemma~\ref{LemTypeEigenvalues}, equivalent to not agreeing on the type of system.

In the current more general communication scenario, where Alice and Bob have \emph{not} equalized their descriptions of the different types,
Alice sends out a classical description of $X_A$ and $\mathcal{P}_A$ and $\rho$. Since Bob uses
a different reference system, he will not know which system and state he should send to Alice. Thus, as in the previous sections, there should be a transformation $T$
(depicted in Figure~\ref{fig_scenario1}) that corrects this description before it arrives in Bob's lab. There is no need to correct the description of $\rho$ -- once
Bob knows the reference system, he also knows the choice of Hilbert space basis and knows how to prepare $\rho$. We will now describe a group of transformations $T$
that achieves the task; these transformations act on $X$ and $\mathcal{P}$ only.

First, consider the case that Alice's and Bob's reference systems have the same parity, in the sense that $\hat M(S_0^B)=Y^\dagger \hat M(S_0^A) Y$
(no transposition) for some fixed invertible matrix $Y$. Furthermore, suppose that $S$ is a system that Bob describes as $S_{X_B}$ (no transposition). Then
\[
   \hat M(S)=X_B^\dagger \hat M(S_0^B) X_B = X_B^\dagger Y^\dagger \hat M(S_0^A) Y X_B = X_A^\dagger \hat M(S_0^A) X_A.
\]
Thus, we get $X_B=Y^{-1} X_A$ and $\mathcal{P}_B=\mathcal{P}_A=+1$. On the other hand, if $S$ is a system that Bob describes as $S_{X_B}^T$, then
\[
   \hat M(S)=X_B^\dagger \hat M(S_0^B)^T X_B = X_B^\dagger \bar Y^\dagger \hat M(S_0^A)^T \bar Y X_B,
\]
where the bar denotes componentwise complex conjugation. Thus, we get $X_B=\bar Y^{-1}X_A$ and $\mathcal{P}_A=\mathcal{P}_B=-1$, and we obtain the
following table defining a correcting transformation $T$:
\begin{eqnarray}
   (X_A,+1)&\stackrel T \mapsto& (Y^{-1}X_A,+1) \nonumber\\
   (X_A,-1)&\stackrel T \mapsto& (\bar Y^{-1} X_A,-1).\label{eqFirstKind}
\end{eqnarray}
Second, consider the case that Alice's and Bob's reference systems have opposite parities. Analogous calculations as above show that the correcting
transformation $T$ has to be chosen slightly differently:
\begin{eqnarray}
   (X_A,+1)&\stackrel T \mapsto& (\bar Y^{-1}X_A,-1) \nonumber\\
   (X_A,-1)&\stackrel T \mapsto& (Y^{-1} X_A,+1).\label{eqSecondKind}
\end{eqnarray}
If we denote the transformations $T$ in~(\ref{eqFirstKind}) by $T^+_{Y^{-1}}$ and those in~(\ref{eqSecondKind}) by $T^-_{Y^{-1}}$ {(mind the inverse)},
we obtain the multiplication table
\[
   T_Y^+ T_Z^+ = T_{YZ}^+,\quad T_Y^- T_Z^- = T_{Y\bar Z}^+,\quad T_Y^- T_Z^+ =T_{Y \bar Z}^-,\quad T_Y^+ T_Z^- =T_{YZ}^-.
\]
As expected, these transformations form a group. Since matrices $X_B$ and $X_B'$ that differ
only by a complex phase describe the same type, we have to identify $Y$ and $e^{i\theta}Y$ for $\theta\in\R$, which yields a quotient group.
If we define the maps $G_Y^+(\hat M):=Y \hat M Y^\dagger$ and $G_Y^-(\hat M):=Y \hat M^T Y^\dagger$
on the self-adjoint matrices $\hat M$, then these maps satisfy the same multiplication rules (with $T$ replaced by $G$), so the group we just obtained
is isomorphic to the group $\g$ containing the $G_Y^+$ and $G_Y^-$.
We have just seen that we can always choose an element $T\in\g$
of this group, and apply it as in Figure~\ref{fig_scenario1} to ensure that Alice's and Bob's communication task succeeds. Thus, we obtain the following statement:
\begin{lemma}
\label{LemPSL}
In the communication scenario described above, the group
\[
   \g=\{\hat M \mapsto Y \hat M Y^\dagger \,\,|\,\, \det Y\neq 0\}\cup \{\hat M \mapsto Y \hat M^T Y^\dagger \,\,|\,\, \det Y\neq 0\}
\]
is achievable. This is the group of linear automorphisms of the cone of positive semidefinite complex $N\times N$ matrices. The subgroup
of implementable transformations is its connected component at the identity, which can also be written
\begin{equation}
   {\rm GL}(N,\C)/{\rm U}(1)\enspace = \enspace \R^+\times {\rm PSL}(N,\C).
   \label{eqConnectedComponent}
\end{equation}
\end{lemma}
That the connected component at the identity, i.e.\ the maps $\hat M \mapsto Y\hat M Y^\dagger$, can be written in the form~(\ref{eqConnectedComponent}) follows from the fact that
every conjugation of this form can also be written as $\hat M \mapsto \lambda Z\hat M Z^\dagger$, with $\lambda>0$ and $\det Z=1$. This subgroup is implementable
because it consists of linear isometries which can in principle be physically implemented according to the postulates of quantum mechanics; on the other hand,
the antilinear conjugations $\hat M \mapsto Y \hat M^T Y^\dagger$ (including the antiunitary maps) cannot.

Lemma~\ref{LemPSL} is the analog of Lemma~\ref{LemQuantumGeneral}. The greater freedom of having different types leads to a larger group $\R^+\times {\rm PSL}(N,\C)$ which contains ${\rm PUA}(N)$ of Lemma~\ref{LemQuantumGeneral} as a subgroup.

\subsection{Relating full laboratories: emergence of the Lorentz group}
\label{Subsec44}
As in Section~\ref{SecO3}, we can ask whether \emph{all} local quantum physics of Alice's and Bob's laboratories{\footnote{Or at least the laboratory systems which are in a connected component of a universal measurability graph.}} can be
related by a single finite-dimensional group element, in the more general setting that some observables are universally measurable on several different quantum systems. Indeed, it turns out that this is possible, as long as we make assumptions analogous to those in Subsection~\ref{SubsecFullLab}.

In Section~\ref{SecO3}, there was an encoding map $\varphi^{(S)}$ for every quantum system $S$; now, however, the encodings for equivalent systems
are related to each other by conjugation according to Lemma~\ref{LemTypeTransform}. Therefore, we have an encoding map $\varphi^{\mathbf{S}}$ for every equivalence class $\mathbf{S}$.
Instead of encoding quantum states, we can also encode observables by using a corresponding map $\left(\varphi^{\mathbf{S}}\right)^\dagger$. It is defined
such that it yields the correct expectation values, ${\rm tr}(\rho_{\rm phys}\hat M_{\rm phys})=\tr\left(\varphi^{\mathbf{S}}(\rho_{\rm phys})\left(\varphi^{\mathbf{S}}\right)^\dagger(\hat M_{\rm phys})\right)$. In contrast to Section~\ref{SecO3},
we now have $\varphi\neq\varphi^\dagger$; indeed, if $\varphi(\rho)=X\rho X^\dagger$ (with possibly a transpose on the $\rho$), then
$\varphi^\dagger(\hat M)=\left(X^\dagger\right)^{-1}\hat M X^{-1}$ (with possibly a transpose on the $\hat M$). Nevertheless, we will in the following drop the $\dagger$ symbol, and simplify the notation by using the symbol $\varphi^{\mathbf{S}}$ also for encodings of observables.

It is easy to see that an analogue of Lemma~\ref{Lem3.3} remains true in the more general setting of quantum systems of different types. To state it, we will say that an observable $\hat M$ can be universally measured on two quantum system equivalence classes $\mathbf{S}$ and $\mathbf{S}'$ if and only if there are quantum systems $S\in\mathbf{S}$ and $S'\in\mathbf{S}'$ such that $\hat M$ can be universally measured on $S$ and $S'$. Since all quantum systems in a single equivalence class have all observables universally measurable, this gives a way to relate the descriptions of observables on any quantum system in $\mathbf{S}$ to those on any quantum system in $\mathbf{S}'$.
\begin{lemma}
\label{LemEncodingTypes}
Suppose that a set of observables $\{\hat M_i\}_{i\in I}$ is universally measurable on $\mathbf{S}$ and $\mathbf{S'}$, and that this set is tomographically complete on $\mathbf{S}'$.
That is, the outcome statistics of all the $\hat M_i$ determine the physical state and type of any system $S'\in\mathbf{S'}$ uniquely. Then there is a protocol that allows Alice and Bob to agree on an encoding $\varphi^{\mathbf{S'}}$ by exchanging only classical information, given that they already agree on an encoding $\varphi^{\mathbf{S}}$. Similarly as in Lemma~\ref{Lem3.3}, the map $\varphi^{\mathbf{S}}\mapsto\varphi^{\mathbf{S'}}$ is continuous with respect to changes of $\varphi^{\mathbf{S}}$ that preserve the set of observables which are universally measurable on $\mathbf{S}$ and $\mathbf{S'}$.
\end{lemma}
Like Lemma~\ref{Lem3.3}, this result can be easily understood in terms of the communication scenario
depicted in Figure~\ref{fig_scenario1}. Suppose Alice and Bob have agreed on a common encoding of states on quantum system $\mathbf{S}$. Then Alice can send Bob
a classical request of the following form: \emph{``Please send me the quantum state on a quantum system in equivalence class $\mathbf{S'}$
which has the following outcome eigenvalues [list of numbers] with the following probabilities [list of numbers] if it is sent into the devices $\hat M_1,\ldots,\hat M_n$.''} In this request, the
devices $\hat M_i$ can be described by sending the matrices $\hat M_i(S_0)$, where $S_0$ is any reference quantum system in $\mathbf{S}$ that Alice and Bob have agreed upon beforehand.
From this, Bob will know what devices he should build, in accordance with their agreement on how to encode $\mathbf{S}$. We know from Lemma~\ref{LemTypeEigenvalues} that the eigenvalues of all observables determine the type of the physical system $S'$ uniquely; tomographic completeness means that the observables $\hat M_i$ are sufficient for this\footnote{We leave it as an open problem here how large the set of $\hat M_i$ has to be to allow for this.}. The outcome probabilities will then uniquely determine the quantum state of the corresponding type, which allows Bob to construct the state and send it to Alice.

There is one subtlety in this protocol: if Alice describes the outcome eigenvalues classically, then does Bob understand this message? After all, we have argued at the beginning of this section that Alice and Bob may disagree on a description of physical scalars, for example if they use different units. However, their agreement on the eigenvalues of $\hat M_i(S_0)$ ``sets a scale'' and fixes the corresponding physical unit.

The \textbf{proof of Lemma~\ref{LemEncodingTypes}} can be obtained by arguing in the Schr\"odinger picture, where the devices $\hat M_i$ correspond
to {\it fixed} observables $\hat M_i(\mathbf{S})$ resp.\ $\hat M_i(\mathbf{S'})$, and states in $\mathbf{S}$ resp.\ $\mathbf{S'}$ are described by unnormalized density matrices.
The proof of Lemma~\ref{Lem3.3} can be copied word by word, as long as every appearance of $S$ is replaced by $\mathbf{S}$ (similarly for $S'$).

In an analogous manner, we can generalize Definition~\ref{DefGraph} and define universal measurability graphs in this setting.
\begin{assumptions}[Universal measurability graph]
\label{AssInteractionGeneral}
Consider the set of all finite-dimensional equivalence classes of quantum systems $\mathbf{S},\mathbf{S'},\mathbf{S''},\ldots$; we regard them as vertices of a graph.
Draw a directed edge from $\mathbf{S}$ to $\mathbf{S'}$ if and only if the situation of Lemma~\ref{LemEncodingTypes} holds, i.e.\ if there
is a set of observables $\{\hat M_i\}_{i\in I}$ which is universally measurable on $\mathbf{S}$ and $\mathbf{S'}$ and tomographically complete on $\mathbf{S'}$.

We assume that there exists a two-level system $\mathbf{S}$ that is a ``root'' of this graph, in the sense that every vertex can be reached
from $\mathbf{S}$ by following directed edges. Furthermore, we assume that no quantum system with a partially preferred choice of encoding is a root of this graph.
\end{assumptions}
It is unclear whether it is even \emph{possible} that the smallest-dimensional root of the universal measurability graph has dimension $d=3$ or more.\footnote{This might be impossible for the following reason: if the universally measurable observables correspond to the Lie algebra of the implementable symmetry group $\mathbb{R}^+\times {\rm PSL}(d,\C)$, cf.\ Lemma~\ref{eqConnectedComponent} (in fact, we conjecture this to follow from our postulates), then this could only be true if there were ``enough'' positivity-preserving (due to~(\ref{eqOrder})) Lie algebra representations in higher dimensions. For $d=2$ this is shown to be possible in Lemma~\ref{LemConsistency}.} We leave it open whether one can show that this is impossible, which would imply that the assumption of a ``qubit root'' can be dropped.

We can argue further along the lines of Subsection~\ref{SubsecFullLab} if we keep our notation in the Schr\"odinger picture.
If $\varphi$ is an arbitrary fixed encoding of the qubit equivalence class $\mathbf{S}$, then every other encoding $\tilde\varphi$ is related to this by a map of the form $X\bullet X^\dagger$ or $X\bullet^T X^\dagger$, with $X$
an invertible complex $2\times 2$-matrix. That is, for every encoding $\varphi^{\mathbf{S}}$, there is some $T\in{\rm PLA}(2)$ (the ``projective linear antilinear group'' of maps as above) such that $\varphi^{\mathbf{S}}=T\circ\varphi$.
In this case, we write $\varphi^{\mathbf{S}}=\varphi^{\mathbf{S}}_T$.

All further steps of the proof of Theorem~\ref{TheQubitNormed}
hold also in this more general case, if every occurrence of ${\rm PUA}(2)$ is replaced by ${\rm PLA}(2)$
(and the reference to the special observables $\lambda\cdot\mathbf{1}$ with $\lambda\in\R$ is removed), yielding that $\g_{\min}={\rm PLA}(2)$. Similarly as in the proof of Theorem~\ref{TheQubitNormed}, we also have to ``bundle together'' those quantum systems $\mathbf{S}'$ and $\mathbf{S}''$ which differ only by the set of  observables that are universally measurable on them and on $\mathbf{S}$ (if those sets are related by conjugation and possibly transposition); Alice and Bob will then not be able to tell them apart by sending classical information only. This is the subtlety that was mentioned in Assumptions~\ref{Ass4.8}. A paradigmatic physical example would again be given by the polarization of photons in different directions of propagation.

The group ${\rm PLA}(2)$ has a simpler description. First, every conjugation $X\bullet X^\dagger$ with $X$ an invertible complex matrix can also
be written in the form $\lambda Z \bullet Z^\dagger$, with $\lambda>0$ and $\det Z=1$. Writing $2\times 2$ Hermitian matrices formally
as $4$-vectors, $\rho=\left(\begin{array}{cc} x_0+x_3 & x_1-i x_2 \\ x_1+i x_2 & x_0-x_3\end{array}\right)$, it is well-known that the conjugation
$\rho\mapsto Z\rho Z^\dagger$ acts like an element of the proper orthochronous Lorentz group~\cite{Wald}. Furthermore, the transposition
acts as a space inversion. Therefore, we obtain the following:

\begin{theorem}
\label{TheQubitUnnormed}
In the scenario above, the minimal group is the orthochronous Lorentz group, together with a scaling factor,
$\g_{\min}=\R^+\times{\rm O}^+(3,1)$. The subgroup of implementable transformations is $\R^+\times{\rm SO}^+(3,1)$, the group of proper orthochronous Lorentz transformations,
times a scaling factor.

Furthermore, if $\mathbf{S}$ is the ``root qubit'' of Assumptions~\ref{AssInteractionGeneral}, and $\mathbf{S'}$ any other quantum system
such that all observables of $\mathbf{S}$ are universally measurable on $\mathbf{S}$ and $\mathbf{S'}$, then $\mathbf{S'}$ carries a projective representation of ${\rm SO}^+(3,1)$; the group elements
act as isometries between different Hilbert spaces.
All other quantum systems $\mathbf{S'}$ carry a projective representation of the subgroup of ${\rm SO}^+(3,1)$ which preserves the observables that are universally measurable on $\mathbf{S}$ and $\mathbf{S'}$.
\end{theorem}

In particular, we have again obtained the relation between Alice's and Bob's laboratories from the relation between their qubit descriptions. In this regard qubits assume the role of fundamental building blocks of their physical world.

It remains to briefly discuss how Alice and Bob can actually agree on the description of qubits. A possible procedure is completely analogous to the one described in Section \ref{sec_qbit} for the simpler scenario. As before, its temporal stability requires the {\it inertial frame condition} to hold. The sole difference, however, is that Alice and Bob will also have to determine the \emph{type} of the qubits (in addition to the state).

In the following subsections, we discuss and interpret this result.

\subsection{Interpretation of the minimal group: scaling and spacetime}\label{interpretation}
If Alice and Bob reside in `distant' laboratories, and have never met before to synchronize their frames of reference, then their descriptions of local
quantum physics will be related by an element $(\lambda,T)$ of the group $\g_{\min}=\R^+\times {\rm O}^+(3,1)$, as we have learned in Theorem~\ref{TheQubitUnnormed}.

How can we interpret the scalar $\lambda>0$? If we set $T=\mathbf{1}$ for a moment, and argue in the Heisenberg picture, then Alice's and Bob's observables will be related
by $\hat M_A=\lambda\hat M_B$. This means that if Alice and Bob measure $\hat M$ on the same physical state, then \emph{the probabilities of the different measurement outcomes
will be the same in both cases, but they will describe the outcomes (which are physical quantities) differently:} Bob will say that the possible outcomes are
$m_1^{(B)},\ldots, m_N^{(B)}$ (the eigenvalues of $\hat M_B$), while Alice will assign the values of outcomes $m_1^{(A)}=\lambda m_1^{(B)},\ldots,
m_N^{(A)}=\lambda m_N^{(B)}$.

This does not come as a great surprise because the condition (\ref{eqOrder}), comparing {\it sizes} of measurement outcomes, introduces a {\it scale} into the game on which Alice and Bob have to agree. But the description of a scale requires {\it units} and Alice and Bob may adhere to different conventions and thus use different \emph{units} for the outcomes of observable $\hat M$, like
``kilometers'' versus ``miles''. Due to~(\ref{eqZero}), both Alice and Bob will agree on the case that the measured value is exactly zero, but for all other values,
there will be a scaling factor between their descriptions due to different choices of units; this is the scalar $\lambda>0$. 

Even if $T\neq\mathbf{1}$, then the form of $\g_{\min}$ as a direct product of groups ensures that \emph{there is always a unique consistent way
of assigning such a scaling factor between any two parties}. That is, Alice and Bob can split off a scalar factor $\lambda_{B\to A}$ such that the resulting transformation $T$ becomes
a Lorentz transformation (which leads to $\lambda_{B\to A}=\lambda$ in the example above), and if Bob and Charlie act likewise to obtain a factor
$\lambda_{C\to B}$, then we will have $\lambda_{C\to B}\lambda_{B\to A}=\lambda_{C\to A}$, which is the consistency condition that one would expect to hold in the case
of differing choices of units. Thus, in a world as described in this section, it is consistent for observers to interpret the relation between laboratories in the following way:
\emph{``The actual relation between any pair of laboratories is given by a Lorentz transformation. Furthermore, parties can of course choose different units,
which gives an unimportant and physically irrelevant factor that relates the descriptions.''}

So what about the Lorentz transformations $T\in{\rm O}^+(3,1)$? It is tempting to interpret those as spacetime transformations, and in the following subsection,
we will show that they indeed correspond to spacetime symmetries in the concrete example of relativistic Stern-Gerlach devices.
But is there any apriori reason to attach this interpretation to our abstractly obtained ${\rm O}^+(3,1)$?

Arguments by Wald~\cite{Wald} suggest an affirmative answer to this question. The condition that different observers residing in some spacetime should be able to translate their respective operational descriptions of physical systems (and, more generally, fields) into one another, requires the isometries
of this spacetime to have a representation on the set of states (classical or quantum) of the physical theory defined on that spacetime. Therefore, if we look at the quantum transformations that relate different observers'
descriptions of physics, then this set of transformations should contain a representation of the spacetime isometries.
Consequently, if we assume that there is some underlying spacetime, our result tells us that we should not be surprised to find
that this spacetime has a symmetry group which is related to ${\rm O}^+(3,1)$.

Of course, the present status of our thought experiment is not sufficient to permit Alice and Bob to identify their ambient spatiotemporal structure uniquely as a Minkowski spacetime (not even locally). After all, this would require a representation of the larger Poincar\'e group involving spacetime translations. Our present scheme is therefore, at this stage, still remote from a successful operational spacetime reconstruction.

Nevertheless, the spacetime structure is ultimately encoded in the {\it relations among the entirety of observers} contained in it (or, rather, it {\it is} the set of all those relations). Indeed, following von Weizs\"acker's ideas~\cite{Weizsaecker}, we may even reverse the present argumentation. One may argue that the very operational \emph{definition}
of a spacetime symmetry is to translate between two equivalent but distinct descriptions of the same physics, as long as these are descriptions of large physical
scenarios (full laboratories) given by distant observers.
In this light, the extraction of the Lorentz group from the informational relation of observers is at the very least suggestive of a deeper connection of our communication game and the ambient spatiotemporal structure.

\subsection{Relativistic Stern-Gerlach measurements}\label{sec_SG}
If our result really has a viable spacetime interpretation as suggested in the previous subsection, then there should already be a concrete relativistic example in which observables (or states) transform precisely as $\hat{M}_B=X^\dagger\,\hat{M}_A\,X$ between reference frames $A$ and $B$, where $X\in\,\rm{SL}(2,\mathbb{C})$. (We drop here the scaling factor $\lambda\in\mathbb{R}^+$ in~(\ref{eqConnectedComponent}) and Theorem~\ref{TheQubitUnnormed} such that $\rm{GL}(2,\mathbb{C})$ is replaced by $\rm{SL}(2,\mathbb{C})$.) Indeed, it turns out that a relativistic formulation of the Stern-Gerlach experiment is of exactly this form, in the WKB-limit of the Dirac theory \cite{Westman1,Westman2,PalmerDiss}. We shall only provide a brief description of this example here.

We start with some formal considerations, and give an intuitive physical description in terms of Stern-Gerlach devices further below.
In the WKB approximation of the Dirac equation, one can derive a momentum-dependent inner product for $\rm{SL}(2,\mathbb{C})$-spinorial qubits from the Dirac current \cite{Westman1,Westman2} $
\langle\psi|\phi\rangle_{p}$
which is linear in the four-momentum $p_\mu$ of the particle carrying the qubit. In this way, one obtains Hilbert spaces $\ch_p\simeq\mathbb{C}^2$ with inner products which depend on the four-momentum $p\equiv p_\mu$. The {\it classical} state of motion $p_\mu$ thus represents the \emph{type} of the qubit as defined in Definition~\ref{def_type}. The set of all inner products, and thus Hilbert spaces $\ch_p$, is related linearly by the set of all spacetime Lorentz boosts, and is invariant under local spatial rotations within the reference frame of a given particle. Spatial rotations, which are represented by unitaries, preserve the momentum, which corresponds to the fact that unitary conjugations preserve the type, cf.\ Definition~\ref{def_type}. Boosts $\Lambda$ act as isometries between Hilbert spaces, $\mathcal{H}_p\to\mathcal{H}_{\Lambda p}$, which corresponds to the fact that the matrices $X$ are isometries as explained in Observation~\ref{Obs1}. Thus, the formalism found in this work is identical to the formalism derived from relativistic quantum mechanics in~\cite{Westman1,Westman2}.

In the WKB-limit, every observable on any $\ch_p$ can be written as
\ba
\hat{M}=M_\mu\hat{\sigma}^\mu,\nn
\ea
where $\sigma^\mu$ is the Pauli-4-vector which acts as the operator, and $M_\mu$ is a {\it classical} Minkowski vector. For instance, a relativistic spin operator arises naturally on each $\ch_p$ from the Pauli-Lubanski vector \cite{Westman1,Westman2}. This classical treatment of the vector arises because we treat the different types of qubits -- and thus the corresponding momenta $p$ -- classically by construction. In particular, our framework does not include superpositions of different types and thus, in this case, of different states of motion. 

But while different observers will employ the same representation of the Pauli-4-vector $\hat{\sigma}^\mu$, they will agree on the description of $M_\mu$ only up to the Lorentz transformation $\Lambda$ translating between their inertial frames. One can check that this yields precisely 
\ba
\Lambda_\mu{}^\nu M_\nu\hat{\sigma}^\mu = X^\dagger M_\mu\hat{\sigma}^\mu X,\label{lalala}
\ea
where $X$ is the $\rm{SL}(2,\mathbb{C})$-transformation corresponding to the Lorentz transformation $\Lambda$. 

This allows us to consider concrete observables associated to a Stern-Gerlach device\footnote{In fact, in order to accommodate all types, i.e., all classical states of motion $p_\mu$, within one universal device, one needs to consider an idealized Stern-Gerlach device into which one can shoot particles from all directions and measure the corresponding deviation and/or acceleration.}. An abstract spin operator (e.g.\ one of the Pauli matrices) alone is not suitable to describe the different deviations that qubits in distinct states of motion experience when running through the inhomogeneous magnetic field, as it does not involve a scale.
By contrast, the gradient of the magnetic field $|B^{\rm RF}_\mu|$, as seen in the rest frame of the particle, constitutes a vector
\[
G_\nu:=\frac\mu m \partial_\nu\sqrt{B^{\rm RF}_\xi B_{\rm RF}^\xi}
\]
which accurately describes the different reactions to the Stern-Gerlach measurement ($\mu$ is the magnetic moment and $m$ the mass of the particle). It represents the four-acceleration that the particle experiences due to the magnetic field. The corresponding operator $\hat{G}=G_\mu\hat{\sigma}^\mu$ can be interpreted as a hybrid observable measuring the acceleration of the particle. It is a hybrid observable because it contains a classical and a quantum part: the quantum part determines the orientation relative to the magnetic field as usual, while the classical part translates this information into a corresponding spacetime acceleration\footnote{Clearly, we are not quantizing the electromagnetic field here such that the operator $\hat{G}$ must be of such a hybrid classical and quantum nature.}.

\begin{figure}[hbt!]
\begin{center}
{\includegraphics[scale=.3]{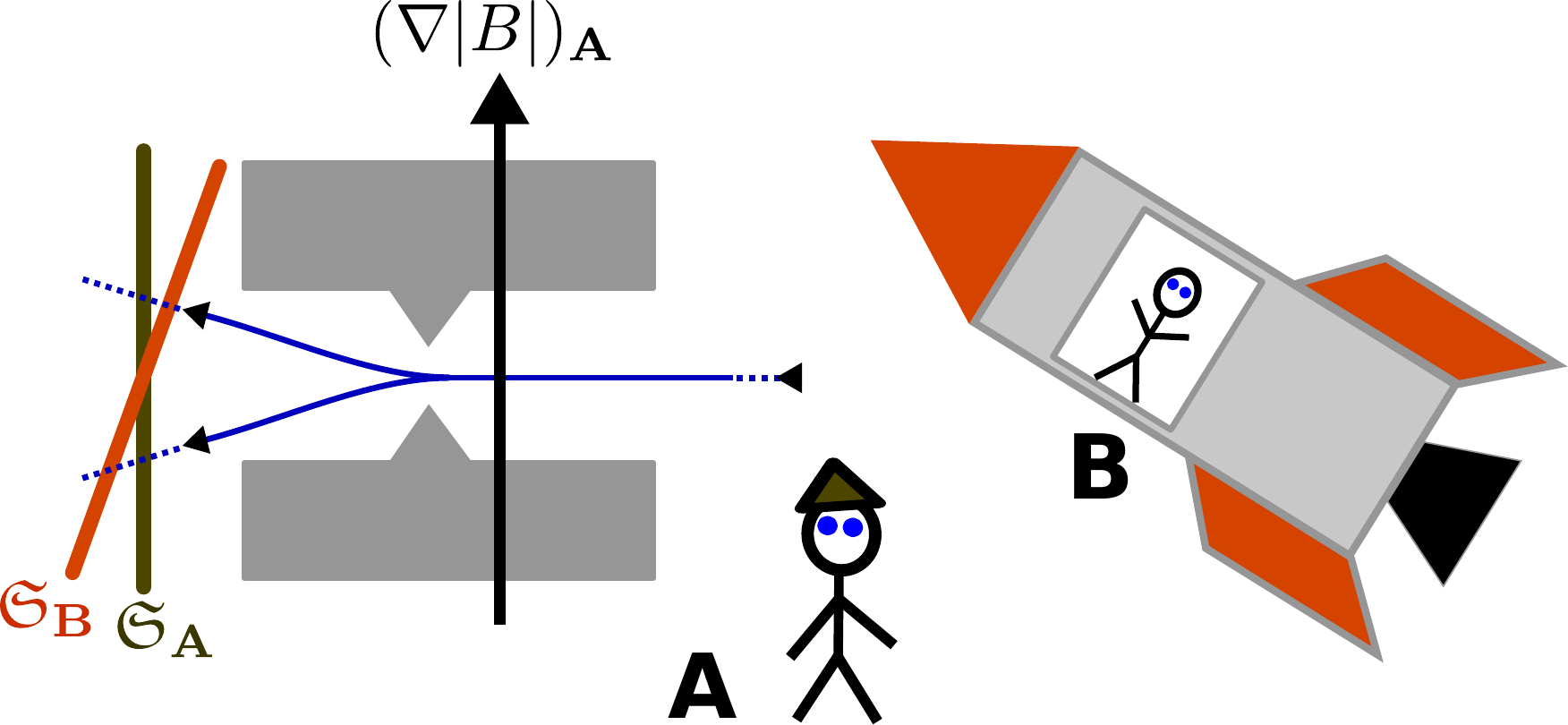}}$\quad$
\includegraphics[scale=0.8]{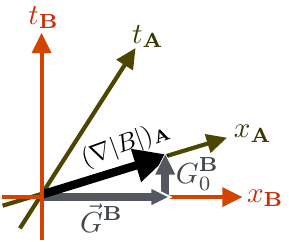}
\caption{\small Alice and Bob, being in different states of relative motion as seen by the injected qubits, will measure the qubits' acceleration according to different ``screens of simultaneity'' $\mathfrak{S}_{\mathbf{A}}$ resp.\ $\mathfrak{S}_{\mathbf{B}}$, and thus see different accelerations and deflections. Bob, having a non-vanishing component along the Stern-Gerlach magnetic field direction, will see `spin up' and `spin down' qubits deflected asymmetrically. The spacetime diagram on the right is drawn from Bob's point of view, assuming that Alice is relative to the qubit at rest.}\label{fig_SG}
\end{center}
\end{figure}

But this $\hat{G}$ is seen differently by different observers. Since the four-acceleration is always orthogonal to the four-momentum \cite{hartlebook}, $p^\mu G_\mu=0$, the acceleration is of the form $G^{\rm RF}=(0,\vec{G})$ in the rest frame of the particle. As one can easily check, this implies that the eigenvalues of the operator $\hat{G}=G_\mu\cdot\hat{\sigma}^\mu$ are $\pm|\vec{G}|$ in the rest frame of the particle, corresponding precisely to the two opposite accelerations of the particles in opposite spin eigenstates. The expectation value of $\hat{G}_{\rm RF}$ as seen from the qubit rest frame corresponds to the average acceleration experienced by the qubit ensemble.

However, Alice and Bob, who observe the deflection of the qubits in the Stern-Gerlach device from different states of motion (see Figure~\ref{fig_SG} for a schematic illustration), will see the resulting acceleration differently, namely as proportional to $\Lambda_\nu{}^\mu G_\mu^{\rm RF}$, where $\Lambda$ is the Lorentz transformation between the qubit's and Alice's or Bob's frame.
Thus, Alice and Bob will see the qubit acceleration in the form $G^{A,B}=(G^{A,B}_0,\vec{G}^{A,B})$, i.e.\ not in general as a four-vector with only spatial components. And this is precisely expressed by the transformation (\ref{lalala}), $\hat{G}_{\rm RF}\mapsto \hat G_{A,B}=X^\dagger \hat G_{\rm RF} X$, which, as can be verified easily, entails the transformed asymmetric eigenvalues $G_0^{A,B}\pm |\vec{G}^{A,B}|$ as seen by Alice or Bob. But these two eigenvalues uniquely determine the acceleration as seen by either Alice or Bob.

The asymmetry corresponds to the fact that the surfaces of simultaneity will be different for Alice, Bob and the qubits. For instance, if Bob's state of motion has non-zero components in the direction of the qubit deflection, one can convince oneself that this necessarily implies that Bob will see the deflection of the qubits asymmetrically. In this setup, one should not think of a Stern-Gerlach device as having a fixed material screen that is hit by the particles after they pass the magnetic field, but of every observer tracking the particle's deflection on his own imagined ``screen of simultaneity''.

Finally, we note that the change of normalization (or non-unitarity of transformations) between different observers corresponds to the fact that what is a purely spatial rotation to one may be seen as a combination of boost and rotation by another. For example, a rotation $R\in\rm{SO}(3)$ of the acceleration $G_\mu$ within the rest frame of the particle amounts to $G^{\rm RF}\mapsto G'^{\rm RF}=(0,R\cdot \vec{G})$ which does {\it not} affect the eigenvalues $\pm|\vec{G}|$. However, if Alice and Bob move relative to the qubits, this transformation will {\it not} appear to them as $G^{A,B}\mapsto G'^{A,B}=(G^{A,B}_0,R\cdot\vec{G}^{A,B})$ because of relativity of simultaneity. Thus, the eigenvalues of the corresponding operator will {\it not} be preserved for Alice and Bob under this transformation. Consequently, what appears as a norm or eigenvalue preserving transformation to one reference frame may constitute a norm and eigenvalue changing transformation for another. 

The unusual `non-unitary' behaviour and the multitude of Hilbert spaces is a consequence of the classical treatment of the {\it type} (i.e., state of motion) of the qubit. Nevertheless, within this framework, everything is consistent and has a clear physical interpretation within a meaningful physical approximation, namely the WKB limit of the Dirac theory. Clearly, for a more fundamental treatment, one would have to consider more sophisticated quantum field theoretic tools where all `types' of systems can be accommodated within one infinite-dimensional Hilbert (or Fock) space. However, the symmetry group does not change under the approximation.

\subsection{Are these results incompatible with a Galilean symmetric world?}

The short answer is: No. None of our assumptions in Definition~\ref{DefUniversal4} (in particular assumptions~(\ref{eqOrder}) and~(\ref{eqZero})) are in themselves
inconsistent with a non-relativistic Galilean universe. However, the details of the answer to this question depend very much on the precise assumptions one is willing to make about the relation between our communication scenario and spacetime, leaving the following two possible hypotheses:
\begin{enumerate}
\item There is no relation between the background spacetime (about which we have not made any specific assumptions) and the quantum states we are describing. Then, by definition, the thought experiment cannot say anything about spacetime at all. In this case, a Galilean world (and, in fact, \emph{any} kind of spacetime) would be consistent with our axioms, and the boost part of our Lorentz group would not represent spacetime boosts.
\item There is actually a relation between the Lorentz group acting on quantum states, obtained in our derivation, and spacetime transformations. Then, if spacetime was Galilean in the Stern-Gerlach example, all inertial observers would agree on the magnitude of the three-acceleration such that an asymmetric deflection of wave packets as in Section~\ref{sec_SG} would not occur. But in this case, a measurement of the acceleration would constitute additional input for Alice and Bob to facilitate their agreement on the description of their quantum states. This would violate Assumptions~\ref{AssInteractionGeneral} by supplementing every quantum system with a partially preferred choice of encoding, bringing us back into the setting of Section~\ref{SecO3}, and yielding ${\mathrm O}(3)$ as the symmetry group.
\end{enumerate}

Given the discussion in Subsection~\ref{interpretation}, and in particular the fact that our abstract assumptions are realized via relativistic Stern-Gerlach devices in the WKB approximation as shown in Subsection~\ref{sec_SG},
we think that there must be some grain of truth to possibility 2. However, from the approach in this paper, we clearly cannot give any conclusive argument as to why the abstractly derived group ${\rm O}^+(3,1)$ must be the spacetime symmetry group, unless we make additional assumptions. The question what would happen in a Galilean background, or some background spacetime of different structure and dimension, warrants further research. The answer to this question will surely depend on the detailed assumptions one makes about how spacetime interacts with the quantum states in our thought experiment.

\section{Summary and conclusions}
We have shown that the Lorentz transformations emerge naturally in an information-theoretic scenario
where two agents in distinct laboratories try to synchronize their descriptions of local quantum physics. The success of synchronization is formally tested in a ``game'', where Alice sends a classical description, and Bob is supposed to send back the physical system that was described. Our derivation does
not assume any specific properties or even the existence of an underlying spacetime manifold (except for very basic assumptions like the possibility to send systems between
observers, or a time order in the local laboratories); instead, the reference frame transformations arise purely from the communication relations among observers as that minimal group which allows them to ``win'' their synchronization game. But it rests on some physical background assumptions. Most importantly,
our result relies on the existence of devices that measure different types of quantum systems {\it universally}, while
satisfying physically meaningful consistency conditions on the physical measurement outcomes.

One possible intuitive understanding of our result is as follows. If we have ``many'' observables on a qubit which can be ``universally measured'' on many other quantum systems (as formalized in the ``graph of universal measurability'',
cf.\ Figure~\ref{fig_graph}), then this allows us to lift any description of qubit states to a description of all quantum systems. This promotes the symmetry group of the qubit
to the full symmetry group that relates different observers' descriptions of quantum physics. This highlights the role of qubits as fundamental informational building blocks in physics -- in line with their importance in reconstructions of quantum theory \cite{Hardy01,Masanes11,Chiribella2011,Hardy11,Axioms2013,Hoehn2014,hw2}. In the simplified case of stand-alone abstract quantum systems,
this yields the orthogonal group ${\rm O}(3)$; in the more realistic case {where we take into account that measurement outcomes on quantum systems actually have a ``size'' that can be compared and carry a scale}, this yields the orthochronous Lorentz group ${\rm O}^+(3,1)$. The resulting formalism correctly
describes relativistic Stern-Gerlach measurements~\cite{Westman2}.

Our result raises several interesting questions. Is there a more direct operational interpretation of the Lorentz transformations, acting on quantum
observables, in terms of relativistic effects like time dilation etc.\ without assuming special relativity in the first place? What if we
have a specific quantum field theory in more than $3+1$ dimensions -- do our operational assumptions or conclusions still
apply to observers in such a universe? Is this approach suitable for information-theoretic descriptions of other
geometric properties like curvature? {Finally, can one conceive of an informational spacetime reconstruction or characterization? For instance, what are the precise information-theoretic conditions from which local observers can infer the existence of a spatiotemporal structure described by Minkowski space?}

Our operational approach {and ``inside point of view'' on spacetime are} in spirit close
to Einstein's original belief, viewing relations among observers as a primary and spacetime as a secondary concept. Namely, the latter is regarded as ``the set of all possible observer relations" rather than an observer independent geometric entity.  This view is also compatible with a recently developed geometric formulation of special and general relativity in terms of an observer space which highlights the observers' relations as fundamental~\cite{Gielen:2012fz}.
The question then arises how far we can push this kind of reasoning towards a deeper information-theoretic understanding
of spacetime. Specifically, it is often stated that quantum theory and general relativity, apart from their apparent conceptual conflicts, are formulated in very different languages -- functional analysis vs.\ geometry. Given the informational character of many structures in quantum theory (as exploited in quantum information theory) and general relativity (as embodied in the causal structure), one might wonder whether information theory can provide a common language for both, preparing the grounds for merging the two theories in a more general informational framework.
In this context, the present paper is supposed to provide a first step towards a novel information-theoretic perspective on spacetime physics: it shows
that there is more than just peaceful coexistence~\cite{Shimony} between
quantum theory and relativity. Instead, the structure of one of these theories tightly constrains the structure of the other.

\section*{Acknowledgments}
We are grateful to Oscar Dahlsten and Robert H\"ubener for many helpful and stimulating discussions and to Ted Jacobson for constructive criticism and comments. We also thank the referees for useful suggestions which helped to improve the presentation of this manuscript.
Research at Perimeter Institute is supported by the Government of Canada through Industry Canada and by the Province of Ontario through the Ministry of
Research and Innovation. The project leading to this publication has also received funding from the European Union's Horizon 2020 research and innovation programme under the Marie Sklodowska-Curie grant agreement No 657661 (awarded to PH). PH would also like to thank the ITP Heidelberg for hospitality during a visit. This research was undertaken, in part, thanks to funding from the Canada Research Chairs program.

\end{document}